%% file: main.tex
\documentclass{article}

\usepackage{arxiv}

\usepackage[utf8]{inputenc} 
\usepackage[T1]{fontenc}    
\usepackage{hyperref}       
\usepackage{url}            
\usepackage{booktabs}       
\usepackage{amsfonts}       
\usepackage{nicefrac}       
\usepackage{microtype}      
\usepackage{lipsum}		
\usepackage{graphicx}
\usepackage{natbib}
\usepackage{doi}

\setcitestyle{square} 
\usepackage{url,hyperref,lineno,microtype,subcaption}
\usepackage[onehalfspacing]{setspace}
\usepackage[final]{changes}
\input{tikz_preamble} 

\usepackage{xcolor}

\usepackage{amsmath,amssymb,amsthm,bm,booktabs,subcaption}

\newtheorem{definition}{Definition}
\newtheorem{theorem}{Theorem}
\newtheorem{lemma}{Lemma}
\renewcommand*{\vec}{\bm}
\newcommand{\dCC}[0]{d_{\rho}}
\newcommand{\louv}[0]{\mathcal{L}}
\newcommand{\parp}[0]{\mathcal{P}}

\newcommand{\sss}[1]{{\rm\scriptscriptstyle{#1}}}
\newcommand{\intra}[0]{\textrm{Intra}}
\newcommand{\inter}[0]{\textrm{Inter}}

\title{The Projection Method: a Unified Formalism for Community Detection}

\author{ \href{https://orcid.org/0000-0002-7197-7682}{\includegraphics[scale=0.06]{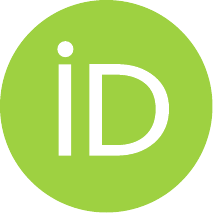}\hspace{1mm}Martijn Gösgens} \\
	Eindhoven University of Technology \\
	\texttt{research@martijngosgens.nl} \\
	\And
	\href{https://orcid.org/0000-0003-1331-9697}{\includegraphics[scale=0.06]{orcid.pdf}\hspace{1mm}Remco van der Hofstad} \\
	Eindhoven University of Technology\\
	 \AND
  \href{https://orcid.org/0000-0002-6750-3484}{\includegraphics[scale=0.06]{orcid.pdf}\hspace{1mm}Nelly Litvak} \\
	Eindhoven University of Technology\\
}

\renewcommand{\headeright}{Gösgens et al.}
\renewcommand{\undertitle}{}
\renewcommand{\shorttitle}{The Projection Method}

\hypersetup{
pdftitle={The Projection Method},
pdfauthor={Martijn Gösgens, Remco van der Hofstad, Nelly Litvak},
pdfkeywords={community detection, clustering, modularity, Markov stability, correlation clustering, granularity},
}

\begin{document}
\maketitle

\begin{abstract}
	We present the class of \emph{projection methods} for community detection that generalizes many popular community detection methods. In this framework, we represent each clustering (partition) by a vector on a high-dimensional \emph{hypersphere}.
A community detection method is a projection method if it can be described by the following two-step approach: 1) the graph is mapped to a \emph{query vector} on the hypersphere; and 2)
the query vector is projected \replaced{on}{to} the set of clustering vectors. This last \emph{projection} step is performed by minimizing the distance between the query vector and the clustering vector, over the set of clusterings.
We prove that optimizing Markov stability, modularity, the likelihood of planted partition models and correlation clustering fit this framework.
A consequence of this equivalence is that algorithms for each of these methods can be modified to perform the projection step in our framework. In addition, we show that these different methods suffer from the same \emph{granularity problem}: they have parameters that control the granularity of the resulting clustering, but choosing these to obtain clusterings of the desired granularity is nontrivial. We provide a general heuristic to address this granularity problem, which can be applied to any projection method.
Finally, we show how, given a generator of graphs with community structure, we can optimize a projection method for this generator in order to obtain a community detection method that performs well on this generator.
\end{abstract}

\keywords{Community detection \and Clustering \and Modularity \and Markov stability \and Correlation clustering \and Granularity}

\section{Introduction}
In complex networks, there \replaced{often are}{are often} groups of nodes that are better connected internally than to the rest of the network. In network science, these groups are referred to as \emph{communities}. These communities often have a natural interpretation: they correspond to friend groups in social networks, subject areas in citation networks, or industries in trade networks. Community detection is the task of finding these groups of nodes in a network. This is typically done by partitioning the nodes, so that each node is assigned to exactly one community. There are many different methods for community detection~\citep{fortunato2010community,fortunato2016community,rosvall2019different}. Yet, it is not easy to say which method is preferable in a given setting.

Most existing methods in network science\deleted{,} detect communities by maximizing a \emph{quality function} that measures how well a clustering into communities fits the network at hand. One of the most widely-used quality functions is modularity~\citep{newman2004finding}, which measures the fraction of edges that are inside communities, and compares this to the fraction that one would expect in a random graph without community structure. One of the main advantages of modularity maximization is that it does not require one to specify the number of communities that one wishes to detect. However, this does not mean that modularity automatically detects the desired number of communities: it is known that in large networks, modularity maximization is unable to detect communities below or above a given size~\citep{Fortunato2007}. This problem is often referred to as the \emph{resolution limit}\added{,} or the \emph{granularity problem}, of modularity maximization. This problem is often ameliorated by introducing a \emph{resolution parameter}~\citep{reichardt2006statistical,traag2011narrow},  which allows \added{one} to control the range of community sizes that modularity maximization is able to detect.

One of the most popular algorithms for modularity maximization is the \emph{Louvain algorithm}~\citep{blondel2008fast}. This algorithm performs a greedy maximization and is able to find a local maximum of modularity in \replaced{linear time (empirically)}{(empirically) linear time} in the number of network edges.

Community detection is closely related to the more general machine learning task of \emph{data clustering}, as we essentially cluster the nodes based on network topology.
In data clustering, the objects to be clustered are typically represented by vectors, and one uses methods like $k$-means~\citep{jain2010data} or spectral clustering~\citep{von2007tutorial} to find a spatial clustering of these vectors so that nearby vectors are assigned to the same cluster.
\added{Community detection can be considered as an instance of clustering, where the elements to be clustered are network nodes.}

In this study, we unify several popular community detection and clustering methods into a single geometric framework. 
We do so by describing a metric space of clusterings, where we represent each clustering $C$ by a binary vector $\vec{b}(C)$ indexed over the node pairs, i.e., $\vec{b}(C)\in\mathbb{R}^{n\choose 2}$, where $n$ is the size of the network. We say that a community detection method is a \emph{projection method} if it is equivalent to the following two-step approach: firstly, the graph is mapped to a \emph{query vector} $\vec{q}\in\mathbb{R}^{{n\choose 2}}$. Secondly, this query vector is \emph{projected} to the set of clustering vectors. That is, we search for the clustering vector $\vec{b}(C)$ that minimizes the distance to $\vec{q}$. 

It turns out that many community detection methods fit this framework. In~\cite{gosgens2023hyperspherical}, we prove that modularity maximization is a projection method. In this work, we additionally show that several other popular community detection methods are projection methods. In Section~\ref{sec:equivalences}, we show that Correlation Clustering~\citep{bansal2004correlation}, the maximization of Markov stability~\citep{delvenne2010stability,lambiotte2014random} and likelihood maximization for several generative models~\citep{avrachenkov2020community} are projection methods. We emphasize that in this paper, we establish equivalences between community detection methods in \replaced{the strictest}{most strict} mathematical sense. As such, our analytical results are much stronger than merely pointing out that methods are similar or related.  Specifically, when we say that two methods are equivalent, we mean that their quality functions $f_1$ and $f_2$ define the \emph{exact} same rankings of clusterings, so that for all clusterings $C_1,C_2$, $f_1(C_1)\geq f_1(C_2)$ holds if and only if $f_2(C_1)\geq f_2(C_2)$. 

Some relations between existing community detection methods were already known~\citep{veldt2018correlation,newman2016equivalence}\replaced{. The novelty of this work is that we unify many community detection methods into a single class of projection methods, and uncover the geometric structure that is baked into each of these methods.}{, but when we formulate these methods as projection methods, we uncover the geometric structure that is baked into them.} \replaced{Furthermore, we demonstrate the following advantages of this geometric perspective.}

\added{Firstly, we show that any community detection method that maximizes or minimizes a weighted sum over pairs of vertices is a projection method. This unifies many well-known methods (Correlation Clustering~\citep{bansal2004correlation}, Markov Stability~\citep{delvenne2010stability}, modularity maximization~\citep{newman2004finding}, likelihood maximization~\citep{avrachenkov2020community}), and any other current or future method that can be presented in this form. Importantly, the hyperspherical geometry comes with natural measures for clustering granularity (the latitude) and the similarity between clusterings (the correlation distance). These measures are additionally related to the quality function (the angular distance) by the \emph{hyperspherical law of cosines}, as we explain in Section~\ref{sec:projection-method}.}

\added{Secondly, this geometric framework yields understanding that all community detection methods that are generalized by the projection method, suffer from  the same  \emph{granularity problem}. That is, these methods require parameter tuning to produce communities of the desired granularity. In Section~\ref{sec:granularity} we use the hyperspherical geometry to derive a general heuristic that addresses this problem. We demonstrate that this heuristic, obtained in our earlier work~\cite{gosgens2023correcting}, can be applied to any projection method.} 

\replaced{Thirdly, projection methods can be combined by taking linear combinations of their query vectors. In Section~\ref{sec:training}, we demonstrate how we can efficiently find a linear combination that performs well in a given setting.}{
This geometric structure allows us to compare the relative positions of the corresponding query vectors. In addition, it allows us to combine different methods by taking linear combinations of their corresponding query vectors, as we will demonstrate in Section~\ref{sec:training}.}

\added{
As a side remark, we note that in network science, the term ``clustering" is also used to refer to the abundance of triangles in real-world networks, which is often quantified by the \emph{clustering coefficient}~\citep{watts1998collective,newman2003properties}.
The presence of communities usually goes hand in hand with an abundance of triangles~\citep{peixoto2022disentangling}. We emphasize that in the present work, we use the term ``clustering" to refer to data clustering, and not to the clustering coefficient.  Nevertheless, the global clustering coefficient can be expressed in terms of our hyperspherical geometry~\citep{gosgens2023hyperspherical}.
}




\subsubsection*{Outline}
The remainder of the paper is organized as follows: in Section~\ref{sec:projection-method}, we describe the projection method and the hyperspherical geometry of clusterings. In Section~\ref{sec:equivalences}, we prove that several popular community detection methods are projection methods and discuss the implications of these equivalences. In Section~\ref{sec:algorithms}, we discuss algorithms that can be used to perform the projection step in the projection method. Finally, Section~\ref{sec:experiments} presents methodology for choosing a suitable projection method in given settings. In particular, Section~\ref{sec:granularity} discusses how to modify a query mapping in order to obtain a projection method that detects communities of desired granularity, while Section~\ref{sec:training} demonstrates how we can can perform \emph{hyperparameter tuning} within the projection method.
Our implementation of the projection method and the experiments of Section~\ref{sec:experiments} is available on Github\footnote{The code is available at \url{https://github.com/MartijnGosgens/hyperspherical_community_detection}.}.

\subsubsection*{Notation}
This paper discusses the relations between several different community detection methods, that each come with their own notations. We aim to keep  notation\deleted{s} as consistent as possible, and here we list \added{the} most common notation\deleted{s} that we will use throughout the paper. 

We represent a graph by an $n\times n$ adjacency matrix $A$ with node set $[n]=\{1,\dots,n\}$ and $m$ edges. We denote the degree of node $i$ by $d_i$. We write $\sum_{i<j}$ to denote a sum over all node pairs $i,j\in[n]$ with $i<j$. We denote the number of node pairs by $N={n\choose 2}$.
For a clustering $C$, we define $\intra(C)$ as the set of intra-cluster node pairs, i.e., pairs of nodes $i<j$ that are part of the same cluster according to $C$. Similarly, we define $\inter(C)$ as the set of node pairs that are part of different clusters according to $C$. We define $m_C=|\intra(C)|$ as the number of intra-cluster pairs.
For two functions $f,g$, we write $f(C)\equiv g(C)$ and say that optimizing $f$ is equivalent to optimizing $g$ if, for each pair of clusterings $C_1,C_2$, the inequality $f(C_1)\geq f(C_2)$ holds if and only if  $g(C_1)\geq g(C_2)$.
We denote vectors by bold letters $\vec{x},\vec{y}$ and we denote the inner product between two vectors by $\langle\vec{x},\vec{y}\rangle$. The Euclidean length of a vector is given by $\|\vec{x}\|=\sqrt{\langle\vec{x},\vec{x}\rangle}$.


\section{The projection method}\label{sec:projection-method}
In this section, we describe the hyperspherical geometry that the projection method relies on. For more details, we refer to~\cite{gosgens2023hyperspherical}.
We consider a graph with $n$ nodes and define a clustering as a partition of these nodes. For a clustering $C$, we define the \emph{clustering vector} $\vec{b}(C)$ as the binary vector indexed by the node-pairs, given by
$$
    \vec{b}(C)_{ij}=\begin{cases}
        1,&\text{ if $i$ and $j$ are in the same cluster},\\
        -1,&\text{ if $i$ and $j$ are in different clusters}.
    \end{cases}
$$
Note that the dimension of this vector is $N={n\choose 2}$, and that the Euclidean length of each clustering vector is $\sqrt{N}$, so that they are all located on a hypersphere of radius $\sqrt{N}$ around the origin. Because of this, it is natural to consider the geometry induced by the \emph{angular distance}, given by 
\begin{equation}\label{eq:angular-dist}
d_a(\vec{x},\vec{y})=\arccos\left(\frac{\langle\vec{x},\vec{y}\rangle}{\|\vec{x}\|\cdot\|\vec{y}\|}\right).
\end{equation}
The vector representation $\vec{b}(C)$, together with the angular distance\added{,} defines a hyperspherical geometry of clusterings.

\subsection*{Clustering granularity and latitude}
\noindent The clustering into a single community corresponds to the all-one vector $\vec{b}(C)=\vec{1}$, while the clustering into $n$ singleton communities corresponds to $\vec{b}(C)=-\vec{1}$. These two vectors form opposite \emph{poles} on the hypersphere. The extent to which a clustering resembles the former or latter is referred to as its \emph{granularity}: \emph{fine-grained} clusterings consist of many small communities, while \emph{coarse-grained} clusterings consist of few and large communities.
We measure clustering granularity by the \emph{latitude} of the clustering vector. For $\vec{x}\in\mathbb{R}^N$, the latitude is defined as $\ell(\vec{x})=d_a(\vec{x},-\vec{1})$. For a clustering vector $\vec{b}(C)$, this is given by $\ell(\vec{b}(C))=\arccos\left(1-2\frac{m_C}{N}\right)$, where $m_C=|\intra(C)|$ is the number of intra-cluster pairs of $C$.
Note that the number of intra-cluster pairs is related to the sum of the cluster sizes: let $s_1,\dots,s_k$ be the sizes of the clusters of $C$. Then
\[
m_C=\sum_{i=1}^k{s_i\choose 2}=\frac{1}{2}\sum_{i=1}^ks_i^2 - \frac{n}{2}.
\]
Thus, quantifying clustering granularity by the latitude is equivalent to quantifying it by the sum of squared cluster sizes.

\subsection*{Parallels and meridians}
\noindent Borrowing more terminology from geography, for $\lambda\in[0,\pi]$, we define the \emph{parallel} $\parp_\lambda$ as the set of vectors with latitude $\lambda$.
In particular, we refer to $\parp_{\pi/2}$ as the \emph{equator}, which corresponds to the set of vectors perpendicular to $\vec{1}$.
For a vector $\vec{x}$ that is not a multiple of $\vec{1}$, we define the \emph{parallel projection} $\parp_\lambda(\vec{x})$ as the projection of $\vec{x}$ onto $\parp_\lambda$, and it is given by
\begin{equation}\label{eq:parp}
\parp_\lambda(\vec{x})=\frac{\sin(\lambda)\cdot\sqrt{N}}{\left\| \vec{x}-\frac{\langle\vec{x},\vec{1}\rangle}{N} \right\|}\cdot\left( \vec{x}-\frac{\langle\vec{x},\vec{1}\rangle}{N} \right)-\cos(\lambda)\cdot\vec{1}.
\end{equation}
Similarly, we define the \emph{meridian} of $\vec{x}$ as the one-dimensional line $\{\parp_\lambda(\vec{x}):\lambda\in(0,\pi)\}$.

\subsection*{Correlation distance and clustering similarity}
\noindent \replaced{For three vectors $\vec{x},\vec{y},\vec{r}$, we can measure the angle on the surface of the hypersphere that the line from $\vec{x}$ to $\vec{r}$ makes with the line from $\vec{y}$ to $\vec{r}$. This angle $\angle(\vec{x},\vec{r},\vec{z})$ is given by the hyperspherical variant of the \emph{law of cosines}:
\begin{equation}\label{eq:hypercosines}
\cos\angle(\vec{x},\vec{r},\vec{y})=\frac{\cos d_a(\vec{x},\vec{y}) - \cos d_a(\vec{x},\vec{r})\cos d_a(\vec{y},\vec{r})}{\sin d_a(\vec{x},\vec{y})\sin d_a(\vec{y},\vec{r})}.
\end{equation}
In particular, when we take $\vec{r}=-\vec{1}$, this angle corresponds to the angle between the meridians of $\vec{x}$ and $\vec{y}$.
}{
For two vectors $\vec{x}$ and $\vec{y}$, one can measure the angle that their meridians make at $-\vec{1}$.} This angle turns out to have an interesting interpretation\added{, as stated in Theorem~\ref{th:th1}:}
\begin{theorem}[\citep{gosgens2023hyperspherical}]
\label{th:th1}
    For two vectors $\vec{x}$ and $\vec{y}$ that are not multiples of $\vec{1}$, the angle that their meridians make is equal to the \emph{arccosine of the Pearson correlation} between $\vec{x}$ and $\vec{y}$.
\end{theorem}
Because of Theorem~\ref{th:th1}, we call this angle the \emph{correlation distance} between $\vec{x}$ and $\vec{y}$. \replaced{Note that $d_a(\vec{x},-\vec{1})=\ell(\vec{x})$, so that the correlation distance is given by}{The correlation distance can be computed from $d_a(\vec{x},\vec{y})$ and the latitudes $\ell(\vec{x}),\ell(\vec{y})$ by the following formula}
\begin{equation}\label{eq:correlation-dist}
    \dCC(\vec{x},\vec{y})=\arccos\left(\frac{\cos d_a(\vec{x},\vec{y})-\cos\ell(\vec{x})\cos\ell(\vec{y})}{\sin\ell(\vec{x})\sin\ell(\vec{y})}\right),
\end{equation}
and the Pearson correlation between vectors $\vec{x}$ and $\vec{y}$ is thus given by $\cos\dCC(\vec{x},\vec{y})$.
The correlation coefficient between two clustering vectors $\vec{b}(C),\vec{b}(T)$ turns out to be a useful quantity for measuring the similarity between clusterings $C$ and $T$~\citep{gosgens2021systematic}.
There exist many measures to quantify the similarity between two clusterings, but most of these measures suffer from the defect that they are biased towards either coarse- or fine-grained clusterings~\citep{vinh2009information,lei2017ground}. The Pearson correlation between two clustering vectors does not suffer from this bias, and additionally satisfies many other desirable properties~\citep{gosgens2021systematic}.
Because of that, we will use the Pearson correlation to measure the similarity between clusterings. For clusterings $C$ and $T$, the Pearson correlation is given by
\begin{equation}\label{eq:pearson-correlation}
\rho(C,T)=\cos\dCC(\vec{b}(C),\vec{b}(T))=\frac{m_{CT}\cdot N-m_C\cdot m_T}{\sqrt{m_C\cdot (N-m_C)\cdot m_T\cdot(N-m_T)}},
\end{equation}
where $m_C=|\intra(C)|, m_T=|\intra(T)|$ and $m_{CT}=|\intra(C)\cap\intra(T)|$.
In cases where $C$ and $T$ correspond to the detected and planted (i.e., ground-truth) clusterings, we use $\rho(C,T)$ as a measure of the performance of the community detection.

\subsection*{Query mappings}
\noindent Above, we have defined the hyperspherical geometry of clusterings. This geometry comes with natural measures for clustering granularity (the latitude $\ell$) and similarity between clusterings (the correlation $\rho$). The idea behind the projection method is that we map a graph to a point in this same geometry, and then find the clustering vector that is closest to that point. For a graph with adjacency matrix $A$, we denote the vector that it is mapped to by $\vec{q}(A)\in\mathbb{R}^N$, and refer to it as the \emph{query vector} of $A$. We refer to $\vec{q}(\cdot)$ as the \emph{query mapping}. The name `query' comes from the fact that in the second step of the projection method, we search for the clustering vector $\vec{b}(C)$ that minimizes $d_a(\vec{q}(A),\vec{b}(C))$. That is, among the set of clustering vectors, we find the one that is \emph{nearest} to $\vec{q}(A)$. In short, we arrive at the following definition of  the projection method:
\begin{definition}
    A community detection method is a \emph{projection method} if it can be described
by the following two-step approach: \added{(}1) the graph with adjacency matrix $A$ is first mapped to a \emph{query vector} $\vec{q}(A)$; and \added{(}2) the query vector is projected to the set of clustering vectors by minimizing $d_a(\vec{q}(A),\vec{b}(C))$ over the set of clusterings $C$.
\end{definition}

There exist infinitely many ways to map graphs to query vectors. One of the simplest way\replaced{s}{,} is to simply turn the adjacency matrix into a vector like $\vec{q}(A)_{ij}=\tfrac{1}{2}\left(A_{ij}+A_{ji}\right)$, where the average is taken in case $A$ is directed. In general, we define the \emph{half-vectorization} of a matrix $X\in\mathbb{R}^{n\times n}$ by
\[
\vec{v}(X)_{ij}=\frac{1}{2}\left(X_{ij}+X_{ji}\right).
\]
The vector $\vec{v}(A^r)$ counts the number of paths of length $r$ between each pair of vertices. In particular, $\vec{v}(A^2)_{ij}$ corresponds to the number of neighbors that the nodes $i$ and $j$ share.

There are many more ways in which we can construct a query vector based on $A$. For example, the entry $\vec{q}(A)_{ij}$ can also depend on the degrees of $i$ and $j$ or even the length of the shortest path between $i$ and $j$.

Finally, because we are minimizing the angular distance, the length of $\vec{q}(A)$ is not relevant. It may be natural to normalize all vectors to a Euclidean length $\sqrt{N}$ so that they have the same length as clustering vectors, but we will not do so to avoid cluttering the notation.

\added{In summary, the hyperspherical geometry comes with three key measures: firstly, the \emph{angular distance} $d_a(\vec{q}(A),\vec{b}(C))$ is the quality measure that we minimize in order to detect communities. Secondly, the \emph{latitude} measures the granularity of a clustering. That is, $\ell(\vec{b}(T))$ measures the granularity of the planted clustering, while $\ell(\vec{b}(C))$ measures the granularity of the detected clustering. Thirdly, the \emph{correlation distance} between the planted and detected communities $\dCC(\vec{b}(C),\vec{b}(T))$ (or its cosine, $\rho(C,T)=\cos\dCC(\vec{b}(C),\vec{b}(T))$) measures the performance of the detection. In Section~\ref{sec:granularity}, we will additionally see that $\dCC(\vec{q}(A),\vec{b}(T))$ is a useful measure. The angular distance, correlation distance and latitude are related by the \emph{hyperspherical law of cosines}, given in~\eqref{eq:correlation-dist}.}



\section{Equivalences to other community detection methods}\label{sec:equivalences}
In this section, we will prove that the class of projection methods generalizes several community detection methods. We prove that the class of projection methods is equivalent to correlation clustering, and we prove that the remaining methods are subclasses of the class of projection methods. For each of the related clustering and community detection methods, we will provide the query mapping of the corresponding projection method. The relations between the methods that are discussed here are illustrated in Figure~\ref{fig:equivalences}.

\begin{figure}
    \centering
    \includegraphics[width=0.6\textwidth]{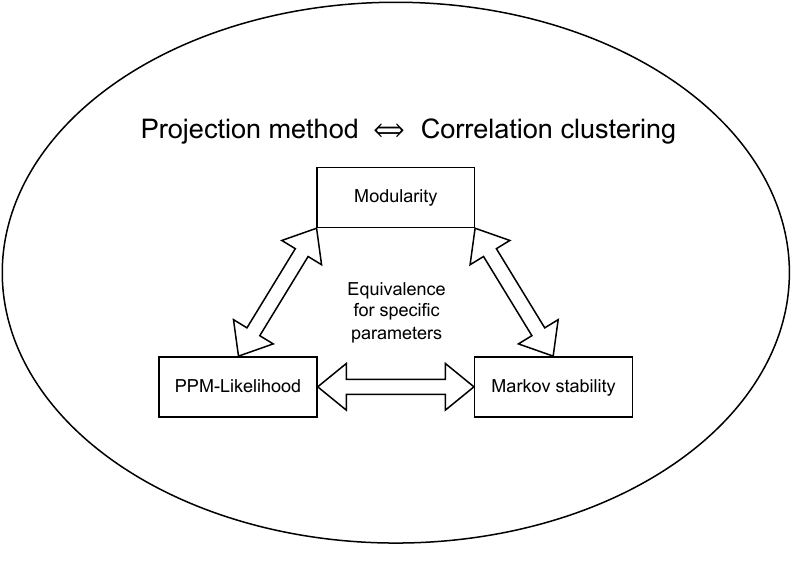}
    \caption{A schematic overview of the equivalences between the community detection and clustering methods that are described in Section~\ref{sec:equivalences}. The projection method is completely equivalent to correlation clustering, while modularity, PPM-likelihood and Markov stability are subsets. For certain parameter choices, these latter three methods are equivalent.}
    \label{fig:equivalences}
\end{figure}

\subsection{Correlation clustering}
Correlation clustering~\citep{bansal2004correlation} is a framework for clustering where pairwise similarity and dissimilarity values $w_{ij}^+$ and $w_{ij}^-$ are given for every pair of objects $i,j$. The objective is to maximize the similarity within the clusters and the dissimilarity between the clusters. Such {\it agreement} of a clustering $C$ with weights $w_{ij}^{\pm}$ is expressed as
\[
\textrm{CorClust}_{\max}(C;w^+,w^-)=
\sum_{ij\in\intra(C)}w_{ij}^++\sum_{ij\in\inter(C)}w_{ij}^-.
\]
Correlation clustering solves the maximization problem 
\begin{equation}
    \label{eq:CorClustMax}
\max_{C}\textrm{CorClust}_{\max}(C;w^+,w^-).    
\end{equation}
Equivalently, one can express the {\it disagreement} of a clustering $C$ with weights $w_{ij}^{\pm}$  as 
\begin{align*}
\textrm{CorClust}_{\min}(\mathcal{C}(w^+,w^-);w^+,w^-)&=\sum_{ij\in\intra(C)}w_{ij}^-+\sum_{ij\in\inter(C)}w_{ij}^+\\
&=\sum_{i<j}(w_{ij}^++w_{ij}^-)-\textrm{CorClust}_{\max}(\mathcal{C}(w^+,w^-);w^+,w^-).
\end{align*}
Then \eqref{eq:CorClustMax} can be stated as a minimization problem
\begin{equation}
    \label{eq:CorClustMin}
\min_{C}\textrm{CorClust}_{\min}(C;w^+,w^-).    
\end{equation}
Somewhat counter-intuitively, the equivalent formulations \eqref{eq:CorClustMax} and \eqref{eq:CorClustMin} lead to different approximation results.  Indeed, for minimization problems \eqref{eq:CorClustMin}, an $\alpha$-approximation guarantee means that for a given  $\alpha> 1$, an optimization algorithm $\mathcal{C}$ satisfies the condition 
\begin{equation}\label{eq:min-approx}
\textrm{CorClust}_{\min}(\mathcal{C}(w^+,w^-);w^+,w^-)\leq \alpha\cdot\textrm{CorClust}_{\min}(C^*;w^+,w^-), \quad \mbox{
for any $w^+,w^-$},\end{equation}
where $C^*$ solves \eqref{eq:CorClustMin}. On the other hand, for the maximization problem \eqref{eq:CorClustMax}, a $\beta$-approximation guarantee means that for a given $\beta<1$,  the optimization algorithm $\mathcal{C}$ satisfies the condition 
\begin{equation}\label{eq:max-approx}
\textrm{CorClust}_{\max}(\mathcal{C}(w^+,w^-);w^+,w^-)\geq \beta\cdot\textrm{CorClust}_{\min}(C^*;w^+,w^-), \quad \mbox{for any 
$w^+,w^-$}.\end{equation}
Now, suppose, we have found an $\alpha$-approximation in \eqref{eq:min-approx}. Then\deleted{, we have}
\begin{align*}
\textrm{CorClust}_{\max}(\mathcal{C}(w^+,w^-);w^+,w^-)
&=\sum_{i<j}(w_{ij}^++w_{ij}^-)-\textrm{CorClust}_{\min}(\mathcal{C}(w^+,w^-);w^+,w^-)\\
&\geq \sum_{i<j}(w_{ij}^++w_{ij}^-)-\alpha\cdot\textrm{CorClust}_{\min}(C^*;w^+,w^-)\\
&=\alpha\cdot\textrm{CorClust}_{\max}(C^*;w^+,w^-)-(\alpha-1)\cdot \sum_{i<j}(w_{ij}^++w_{ij}^-),
\end{align*}
which gives \emph{some} approximation guarantee for the maximization problem, but not  of the same multiplicative form as~\eqref{eq:max-approx}.

The motivation for correlation clustering originates from the setting where we are given a noisy classifier that, for each pair of objects, predicts whether they should be clustered together or apart~\citep{bansal2004correlation}. This leads to the simple $\pm 1$ version of correlation clustering, where  $w_{ij}^+,w_{ij}^-\in\{0,1\}$ and $w_{ij}^++w_{ij}^-=1$ (or equivalently, $w_{ij}^+-w_{ij}^-=\pm1$) holds for each pair $ij$.
The best known approximation guarantee for the minimization formulation of the $\pm1$ variant is $2.06$, which is achieved by rounding the solution from an LP relaxation~\citep{chawla2015near}.

For the general case where the weights are unconstrained, it is known that maximizing agreement is APX-hard~\citep{charikar2005clustering}, which means that any constant-factor approximation is NP-hard. The same authors do provide a 0.7666-approximation algorithm by rounding the semi-definite programming solution. Note that this does not contradict the APX-hardness result, as semi-definite programming is also NP-hard.

We now prove that correlation clustering corresponds to a projection method:

\begin{lemma}\label{lem:correlation-clustering}
    Correlation clustering with similarity and dissimilarity values $w^+=(w_{ij}^+)_{i<j},w^-=(w_{ij}^-)_{i<j}$ is equivalent to a projection method with query vector $\vec{q}^\sss{(CC)}$ given by
    \[
        \vec{q}^\sss{(CC)}_{ij}=w_{ij}^+-w_{ij}^-.
    \]
\end{lemma}
\begin{proof}
    Minimizing $d_a(\vec{q}^\sss{(CC)},\vec{b}(C))$ is equivalent to maximizing $\langle \vec{q}^\sss{(CC)},\vec{b}(C)\rangle=\|\vec{q}^\sss{(CC)}\|\cdot\sqrt{N}\cdot\arccos d_a(\vec{q}^\sss{(CC)},\vec{b}(C))$. We prove that this, in turn, is equivalent to  maximizing $\textrm{CorClust}_{\max}(C;w^+,w^-)$  w.r.t. to the given values $(w_{ij}^+)_{i<j}$ and $(w_{ij}^-)_{i<j}$:
    \begin{align*}
        \langle \vec{q}^\sss{(CC)},\vec{b}(C)\rangle
        &=\sum_{ij\in\intra(C)}(w_{ij}^+-w_{ij}^-)+\sum_{ij\in\inter(C)}(w_{ij}^--w_{ij}^+)\\
        &=2\cdot\left(\sum_{ij\in\intra(C)}w_{ij}^++\sum_{ij\in\inter(C)}w_{ij}^-\right)-\sum_{i<j}(w_{ij}^++w_{ij}^-)\\
        &=2\cdot\textrm{CorClust}_{\max}(C;w^+,w^-)-\sum_{i<j}(w_{ij}^++w_{ij}^-).
    \end{align*}
    The last term does not depend on $C$, so we obtain that indeed maximizing $\langle \vec{q}^\sss{(CC)},\vec{b}(C)\rangle$ is equivalent to maximizing $\textrm{CorClust}_{\max}(C;w^+,w^-)$, as required.
\end{proof}

It is easy to see that any query vector $\vec{q}$ can also be turned into a correlation clustering objective by taking $w_{ij}^+=\max\{0,\vec{q}_{ij}\}$ and $w_{ij}^-=\max\{0,-\vec{q}_{ij}\}$. This tells us that correlation clustering and projection methods are \emph{equivalent}, in the sense that any correlation clustering instance can be mapped to a query vector and vice versa. Note, however, that both classes come with different invariances:  correlation clustering is invariant to applying the same linearly increasing transformation $f(x)=a+bx$ to all values $w_{ij}^+$ and $w_{ij}^-$, where $a,b\in\mathbb{R}$ and $b>0$.
Similarly, projection methods are invariant to multiplying the query vector by any positive constant, due to the hyperspherical geometry.

The equivalence between correlation clustering and projection methods gives a new interpretation for correlation clustering in terms of hyperspherical geometry, and allows one to transfer hardness results from correlation clustering to projection methods.

\subsubsection*{Correlation clustering vs. community detection\deleted{.}}
While the fields of community detection and correlation clustering are similar, the focus of the two fields is notably different: correlation clustering studies clustering from an algorithmic viewpoint, where the goal is to design algorithms with provable optimization guarantees with respect to the correlation clustering quality function. In contrast, community detection focuses mainly on the choice of the quality function, with the aim to obtain a meaningful clustering into communities. In this context, `meaningful' can mean either statistically significant, or similar to some ground truth clustering. In summary, community detection asks ``What quality function to optimize?" while correlation clustering asks ``What algorithm is best for optimizing the correlation clustering quality function?".

\cite{veldt2018correlation} introduced an interesting variant of correlation clustering methods known as \emph{LambdaCC}, and showed that it is related to Sparsest Cut, Normalized Cut and Cluster Deletion.
They additionally introduce a degree-corrected variant of LambdaCC and prove that it is equivalent to CL-modularity, which we discuss in Section~\ref{sec:modularity}.

\subsection{Markov stability}
The \emph{Markov stability}~\citep{delvenne2010stability} of a clustering with respect to a network quantifies how likely a random walker is to find itself in the same community at the beginning and end of some time interval. When communities are clearly present in a network, then a random walker will tend to stay inside communities for long time periods and travel between communities infrequently. Markov stability can be defined for various types of discrete- or continuous-time random walks~\citep{lambiotte2014random}. 
Let $P(t)_{ij}$ be the probability that the random walker is at node $j$ at time $t$ if it was at node $i$ at time $0$, and denote by $P(t)=(P(t)_{ij})\in\mathbb{R}^{n\times n}$ the $t$-transition matrix of the random walk. For simplicity, in this paper, we only consider discrete time $t=0,1,\ldots$, thus $P(t)=P^t$, where $P=P(1)$. We assume that $P$ is irreducible and aperiodic, so that the random walk has a unique stationary distribution $\vec{s}\in\mathbb{R}^{n}$, which is the unique solution to $\vec{s}=\vec{s}P$, such that all elements of $\vec{s}$ sum up to one. We consider a random walker starting from the initial state sampled from $\vec{s}$, and compare the distribution of its location at time $t$, to another location sampled from the stationary distribution. Markov stability measures the covariance between the community label indicators before and after the interval $t$.

Formally, let $C$ be a clustering and let the clusters be numbered $1,2,\dots,k$, where $k$ is the number of clusters in $C$.
We denote by $H(C)$ the $n\times k$ indicator matrix of the clustering $C$, where $H(C)_{ia}=1$ if node $i$ belongs to the $a$-th cluster, and $H(C)_{ia}=0$ otherwise. 
Markov stability is defined as
\begin{equation}
\label{eq:MarkovStability}
\text{MarkovStability}(C,P(t))=\text{Trace}\left( 
    H(C)^\top \left(\text{diag}(\vec{s})P(t)-\vec{s}^\top \vec{s}\right)H(C)
\right).    
\end{equation}

For more details on the definition of Markov stability and its variants, we refer to~\cite{delvenne2010stability} and~\cite{lambiotte2021modularity}.
We show that maximizing Markov stability is a projection method with respect to the query vector
\begin{equation}\label{eq:markov-vector}
    \vec{q}^\sss{(MS)}_t=\vec{v}(\text{diag}(\vec{s})P(t)-\vec{s}^\top \vec{s}),
\end{equation}
where $\vec{v}(X)\in\mathbb{R}^N$ is the \emph{half-vectorization} of the matrix $X$, defined by $\vec{v}(X)_{ij}=\tfrac{1}{2}(X_{ij}+X_{ji})$ for $i<j$. We show that for \emph{any} matrix $X\in\mathbb{R}^{n\times n}$, maximizing $\text{Trace}(H(C)^\top X H(C))$  is a projection method\replaced{:}{.}
\begin{lemma}\label{lem:trace}
    For any matrix $X\in\mathbb{R}^{n\times n}$, maximizing the trace
    \begin{equation}\label{eq:trace-maximization}
        \text{Trace}(H(C)^\top X H(C))
    \end{equation}
    over the set of all clusterings $C$ is equivalent to minimizing 
    $
    d_a\left(\vec{v}(X),\vec{b}(C)\right)
    $.
\end{lemma}
\begin{proof}
    The trace is written as 
    \[
    \text{Trace}(H(C)^\top X H(C))=\sum_{a}\left[H(C)^\top X H(C)\right]_{aa}.
    \]
    We write 
    \begin{equation}
    \label{eq:MarkovStability-1}
    \left[H(C)^\top X H(C)\right]_{aa}=
    \sum_i\sum_j H(C)_{ia}X_{ij}H(C)_{ja}.    
    \end{equation}    
    Now, note that $H(C)_{ia}H(C)_{ja}=1$ whenever $i$ and $j$ are both in community $a$, and $H(C)_{ia}H(C)_{ja}=0$ otherwise.
    Therefore, summing over $a$, we get $\sum_aH(C)_{ia}H(C)_{ja}=1$ if $ij\in\intra(C)$, $ji\in\intra(C)$ or $i=j$.
    Hence, 
    \[
    \text{Trace}(H(C)^\top X H(C))=\sum_{a}\left[H(C)^\top X H(C)\right]_{aa}=\sum_{ij\in\intra(C)}(X_{ij}+X_{ji})+\sum_{i\in[n]}X_{ii}.
    \]
    Note that the sum $\sum_{i\in[n]}X_{ii}$ does not depend on $C$, so that omitting it will not affect the optimization. In addition, we can subtract $\tfrac{1}{2}\sum_{i<j}(X_{ij}+X_{ji})$, which is also constant w.r.t. $C$. We obtain
    \[
    \text{Trace}(H(C)^\top X H(C))\equiv \sum_{ij\in\intra(C)}\tfrac{1}{2}(X_{ij}+X_{ji})-\sum_{ij\in\inter(C)}\tfrac{1}{2}(X_{ij}+X_{ji})=\langle\vec{v}(X),\vec{b}(C)\rangle.
    \]
    To conclude, trace maximization is equivalent to maximizing $\langle\vec{v}(X),\vec{b}(C)\rangle$, which is equivalent to minimizing $d_a(\vec{v}(X),\vec{b}(C))$ over the set of clusterings $C$.
\end{proof}
It is known~\citep{delvenne2010stability} that for a discrete-time Markov chain and $t=1$, Markov stability is equivalent to CL-modularity maximization with $\gamma=1$\added{, which we define in Section~\ref{sec:modularity}}.
The time parameter $t$ controls the granularity of the detected communities. 
When $t=0$, we get communities of size $1$, because $P^0=I$, so that $\left(\text{diag}(\vec{s})P^t-\vec{s}^\top \vec{s}\right)_{ij}=-\vec{s}_i\vec{s}_j<0$ for all $i\neq j$, therefore \eqref{eq:MarkovStability} is maximized when in \eqref{eq:MarkovStability-1} we have $H_{ia}H_{ja}=0$ for all $a$. 
Furthermore, \cite{delvenne2010stability} show that in the limit $t\to\infty$, Markov stability in continuous time and with normalized Laplacian instead of the matrix $P$, divides the network in two communities, corresponding to the positive and the negative coordinates of the \emph{Fiedler vector}, that is the eigenvector corresponding to the second smallest eigenvalue of the normalized Laplacian.

\subsubsection*{Matrix vs. vector representation}
Note that the vector and matrix representations of clusterings are related by $\vec{b}(C)=2\vec{v}(H(C)H(C)^\top)-\vec{1}$.
This relation makes it possible to re-define the hyperspherical geometry from Section~\ref{sec:projection-method}\deleted{,} entirely in terms of $n\times n$ matrices instead of ${n\choose 2}$-dimensional vectors. However, we refrain from doing so, because we believe that in most cases, the vector representations are easier to work with. To illustrate, the matrix formulation of projection methods\deleted{,} amounts to replacing the \emph{query vector} $\vec{q}$ with a \emph{query matrix} $Q\in\mathbb{R}^{n\times n}$. Note that $Q$ has $n^2$ entries $i,j\in [n]$,  while $\vec{q}$ has only ${n\choose 2}$ entries $i<j$, but these extra entries \replaced{do not}{don't} carry any additional information. Indeed, concerning the off-diagonal elements $i>j$, it does not matter whether $Q$ is symmetric or not, because $\text{Trace}(H(C)^\top Q H(C))=\tfrac{1}{2}\text{Trace}(H(C)^\top (Q+Q^\top) H(C))$ for any $Q\in\mathbb{R}^{n\times n}$. Furthermore, the diagonal entries $i=j$ play no role because adding any diagonal matrix $\text{diag}(\vec{y})$, for $\vec{y}\in \mathbb{R}^n$,  to $Q$, does not affect the optimization:
\[
\text{Trace}(H(C)^\top (Q+\text{diag}(\vec{y})) H(C))=\text{Trace}(H(C)^\top Q H(C))+\sum_{i\in[n]}\vec{y}_i\equiv\text{Trace}(H(C)^\top Q H(C)). 
\]
The vector representation has the advantage that it omits these unimportant values, which is why we use query vectors instead of query matrices. Nevertheless, the matrix representation allows for an easier analysis in some settings. For example, in~\cite{liu2018geometric}, the spectral properties of Markov stability are leveraged to create an optimization algorithm.

\subsection{Modularity}\label{sec:modularity}
Modularity maximization is one of the most widely-used community detection methods~\citep{newman2004finding}. Modularity measures the excess of edges inside communities, compared to a null model; a random graph model without community structure. This null model is usually either the Erd\H{o}s-R\'{e}nyi (ER) model or the Chung-Lu (CL, \cite{chung2001diameter}) model. Modularity comes with a \emph{resolution parameter} that controls the granularity of the detected clustering. For the ER and CL null models, modularity is given by
\begin{align*}
    \textrm{CLM}(C;A,\gamma)=\frac{1}{2m}\sum_{ij\in\intra(C)}\left(A_{ij}-\gamma\frac{d_id_j}{2m}\right),\quad
    \textrm{ERM}(C;A,\gamma)=\frac{1}{2m}\sum_{ij\in\intra(C)}A_{ij}-\gamma\frac{m}{N},
\end{align*}
where we recall that $d_i$ is the degree of node $i$, $m=\tfrac{1}{2}\sum_{i\in[n]}d_i$ is the number of edges, and $N={n\choose 2}$ is the number of node pairs.

In~\cite{gosgens2023hyperspherical}, we have proven that modularity maximization is a projection method.
In addition, the equivalence between modularity maximization and correlation clustering was already proven by~\cite{veldt2018correlation}. Hence, Lemma~\ref{lem:correlation-clustering}, too, establishes that modularity maximization is a projection method. Finally, \cite{newman2006finding} shows that modularity can be written in a similar trace-maximization form as~\eqref{eq:trace-maximization}, which additionally allows one to use Lemma~\ref{lem:trace} to prove that modularity maximization is a projection method. Either way, we get the following query vectors:
\[
    \vec{q}^\sss{(Mod)}_\sss{CL}(A,\gamma) = \vec{v}(A)-\gamma\cdot\vec{d}(A),\quad
\vec{q}^\sss{(Mod)}_\sss{ER}(A,\gamma) = \vec{v}(A)-\gamma\frac{ m}{N} \vec{1},
\]
where $\vec{v}(A)$ is the \emph{adjacency vector} (the half-vectorization of the adjacency matrix); $m$ is the number of edges; and $\vec{d}(A)_{ij}=\tfrac{1}{2m}d_id_j$. \added{For a fixed null model, the vectors that are obtained by varying $\gamma$ have the following geometric interpretation~\citep{gosgens2023hyperspherical}:}
\replaced{b}{B}ecause the query vector $\vec{q}^\sss{(Mod)}_\sss{ER}$ is a linear combination of $\vec{v}(A)$ and $\vec{1}$\deleted{,} with coefficients depending on $\gamma$, it can be shown that $\vec{q}^\sss{(Mod)}_\sss{ER}(A,\gamma)$ lies on the meridian of $\vec{v}(A)$ for every value of $\gamma$. The latitude of this modularity vector is related to the resolution parameter by
\begin{equation}\label{eq:erm-latitude}
    \tan\ell\left(\vec{q}^\sss{(Mod)}_\sss{ER}(A,\gamma)\right)=\frac{\sqrt{\frac{N-m}{m}}}{\gamma-1}.
\end{equation}
The vector $\vec{q}^\sss{(Mod)}_\sss{CL}(A,\gamma)$ does not lie on a single meridian and its latitude does not allow for such an elegant formula. However, the set $\left(\vec{q}^\sss{(Mod)}_\sss{CL}(A,\gamma)\right)_{\gamma\geq0}$ does correspond to a \emph{geodesic} on the hypersphere, which is the spherical analogue to a straight line in Euclidean geometry. This geodesic runs from $\vec{v}(A)$ to $-\vec{d}(A)$.
\added{In general, let $\mathcal{N}$ be a null model and let $\vec{p}^{\mathcal{N}}(A)$ be the corresponding vector of expected edge probabilities for the null model $\mathcal{N}$, then the set $(\vec{q}^\sss{(Mod)}_{\mathcal{N}}(A,\gamma))_{\gamma\in[0,\infty]}$ corresponds to the geodesic between $\vec{v}(A)$ and $-\vec{p}^{\mathcal{N}}(A)$.
In Figure~\ref{fig:modularity_lines}, we illustrate the geodesics formed by the modularity vectors for the ER and CL null models.}

\begin{figure}
    \centering
    \resizebox{0.4\textwidth}{!}{\input{modularity_vectors_drawing}}
    \caption{Illustration of the geodesics formed by varying the resolution parameter of the modularity vector for a fixed null model. The ER-modularity vectors lie on a single meridian, in contrast to the CL-modularity vectors.}
    \label{fig:modularity_lines}
\end{figure}
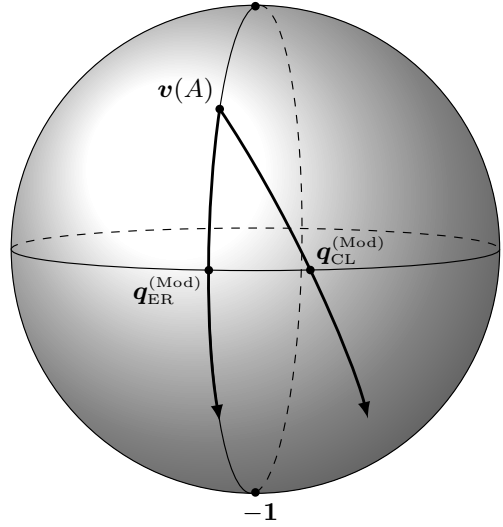

\subsubsection*{The granularity problem of modularity maximization\deleted{.}}
\added{It is well-known that modularity maximization is unable to detect communities that are either too large or too small~\citep{Fortunato2007}. This behavior is often demonstrated by the \emph{ring of cliques}, a graph consisting of $k$ cliques of size $s$ each, where each clique is connected to the next one by a single edge. Let us denote the adjacency matrix of this ring of cliques by $A_{k,s}$, and let $T_{k,s}$ denote its natural clustering into $k$ cliques. For fixed $\gamma$ and $s$ and sufficiently large $k$, modularity-maximizing algorithms will merge neighboring cliques. Geometrically, this problem can be understood as follows: because $N=\Theta(k^2)$ and $m=\Theta(k)$, \eqref{eq:erm-latitude} tells us that $\tan\ell\left(\vec{q}^\sss{(Mod)}_\sss{ER}(A_{k,s},\gamma)\right)=\Theta(\sqrt{k})$ if $\gamma\ne 1$ and $\tan\ell\left(\vec{q}^\sss{(Mod)}_\sss{ER}(A_{k,s},\gamma)\right)=+\infty$ if $\gamma=1$. Thus, for any fixed $\gamma>0$, we have $\ell\left(\vec{q}^\sss{(Mod)}_\sss{ER}(A_{k,s},\gamma)\right)\rightarrow \tfrac{\pi}{2}$ as $k\to\infty$. However, the clustering $T_{k,s}$ has $|\intra(T_{k,s})|=\Theta(k)$, so that $\ell(\vec{b}(T_{k,s}))=\arccos\left(1-2\cdot |\intra(T_{k,s})|/N\right)\rightarrow\arccos(1)=0$. Now, by the inverse triangle inequality, we have
\[
d_a(\vec{q}^\sss{(Mod)}_\sss{ER},\vec{b}(T_{k,s}))\ge |d_a(\vec{q}^\sss{(Mod)}_\sss{ER},-\vec{1})-d_a(-\vec{1},\vec{b}(T_{k,s}))|=|\ell(\vec{q}^\sss{(Mod)}_\sss{ER})-\ell(\vec{b}(T_{k,s}))|\rightarrow \frac{\pi}{2}.
\]
Note that half of the hypersphere lies within a distance $\frac{\pi}{2}$ from the vector $\vec{q}^\sss{(Mod)}_\sss{ER}$. And indeed, it can be shown that the clustering consisting of pairs of neighboring cliques is closer to $\vec{q}^\sss{(Mod)}_\sss{ER}$ than $\vec{b}(T_{k,s})$, so that modularity-maximizing algorithms will prefer such clusterings over the natural clustering into cliques. }

\added{Note that by~\eqref{eq:erm-latitude}, the latitude of the ER-modularity vector is a monotone function of the resolution parameter $\gamma$. Thus, choosing the resolution parameter is equivalent to setting the latitude of the query vector $\lambda=\ell(\vec{q}^\sss{(Mod)}_\sss{ER})$.
The hyperspherical geometry suggests two ways to choose $\lambda$. 
The first approach is to choose $\lambda$ to minimize the angular distance $d_a(\vec{q}^\sss{(Mod)}_\sss{ER},\vec{b}(T_{k,s}))$. Note that changing the query latitude $\lambda$ does not affect the correlation distance $\dCC(\vec{q}^\sss{(Mod)}_\sss{ER},\vec{b}(T_{k,s}))$. This allows us to use the law of cosines~\eqref{eq:hypercosines}, to express $\cos d_a(\vec{q}^\sss{(Mod)}_\sss{ER},\vec{b}(T_{k,s}))$ as a function of $\lambda$, $\lambda_T=\ell(\vec{b}(T_{k,s}))$, and $\theta=\dCC(\vec{q}^\sss{(Mod)}_\sss{ER},\vec{b}(T_{k,s}))$ as
\[
\cos d_a(\vec{q}^\sss{(Mod)}_\sss{ER},\vec{b}(T_{k,s}))=\cos\lambda \cdot \cos\lambda_T+\cos\theta \cdot \sin\lambda \cdot \sin\lambda_T.
\]
Minimizing $d_a(\vec{q}^\sss{(Mod)}_\sss{ER},\vec{b}(T_{k,s}))$ as a function of $\lambda$ yields $\tan \lambda= \cos\theta \tan \lambda_T,$ i.e., $\lambda\sim \cos(\theta)\lambda_T$ as $\lambda_T\to 0$.
The second approach for choosing $\lambda$ is to simply equate the query latitude to the latitude of $T_{k,s}$, so $\lambda=\lambda_T$. 
Both approaches yield $\lambda\rightarrow0$ as $k\to\infty$, so that by the triangle inequality the distance between $T_{k,s}$ and the modularity vector with latitude $\lambda$ vanishes as
\[
d_a(\vec{q}^\sss{(Mod)}_\sss{ER},\vec{b}(T_{k,s}))\le d_a(\vec{q}^\sss{(Mod)}_\sss{ER},-\vec{1})+d_a(-\vec{1},\vec{b}(T_{k,s}))=\lambda+\lambda_T\rightarrow 0.
\]
In terms of the resolution parameter $\gamma$, both these approaches yield $\gamma=\Theta(k)$. Moreover,
it can be shown that for both of these approaches, merging neighboring cliques does not decrease $d_a(\vec{q}^\sss{(Mod)}_\sss{ER},\vec{b}(T_{k,s}))$. This demonstrates that these two approaches effectively address the granularity problem for the ring of cliques.
Let us note that both these approaches heavily rely on the hyperspherical geometry, and that they cannot be applied to the modularity function directly. Indeed, the first approach does not work because maximizing $\text{ERM}(T;A,\gamma)$ as a function of $\gamma$ yields the trivial solution $\gamma=0$, while the second approach is ill-defined without the notion of latitude.}

\added{While these two approaches effectively address the granularity problem of modularity for this ring of cliques network, they do not work well in general. We provide a better and more universal approach  in Section~\ref{sec:granularity}.}



\subsection{Likelihood of generalized planted partition models}\label{sec:likelihood}
The Planted Partition Model (PPM) is one of the simplest random graph models that incorporates community structure.
In this model, we assume there is some planted clustering into communities (the ground truth partition) and that nodes of the same community are connected with probability $p_\sss{in}$, while nodes of different communities are connected with probability $p_\sss{out}<p_\sss{in}$.
The likelihood of the PPM was derived in~\cite{holland1983stochastic}. For an adjacency matrix $A$, the likelihood that it was generated by a PPM with clustering $C$, is given by
\[
\text{PPM-Likelihood}(A|C,p_\sss{in},p_\sss{out})=\prod_{ij\in\intra(C)}p_\sss{in}^{A_{ij}}(1-p_\sss{in})^{1-A_{ij}}\prod_{ij\in\inter(C)}p_\sss{out}^{A_{ij}}(1-p_\sss{out})^{1-A_{ij}}.
\]
In this standard PPM, we see that the adjacency matrix has binary entries.
In~\cite{avrachenkov2020community}, the PPM is generalized to allow for \emph{pairwise interactions} in any measurable set $\mathcal{I}$. That is, we require $A_{ij}\in\mathcal{I}$ for all $i<j$. For example, by taking $\mathcal{I}=\mathbb{R}$, we get a weighted undirected graph, while directed graphs can be modeled by $\mathcal{I}=\mathbb{R}^2$, so that each interaction\deleted{s} corresponds to the tuple of weights $w_{ij}$ and $w_{ji}$. Furthermore, this generalization also allows one to model temporal and multilayer networks.
Similarly to the binary PPM, it is assumed that the distribution of the interaction between $i$ and $j$ only depends on whether $i$ and $j$ are in the same community.  That is, there are likelihood functions $f_\sss{in},f_\sss{out}$ that measure the likelihood of an interaction $A_{ij}\in\mathcal{I}$ resulting from an intra- or inter-community interaction. Importantly, we require the interactions to be pairwise independent. This allows us to express the likelihood of the interaction matrix $A$ as the product of these pairwise likelihoods:
\[
\text{GenPPM-Likelihood}(A|C,f_\sss{in},f_\sss{out})=\prod_{ij\in\intra(C)}f_\sss{in}(A_{ij})\prod_{ij\in\inter(C)}f_\sss{out}(A_{ij}).
\]
After taking the logarithm, it is easy to see that this is equal to the maximization variant of correlation clustering with $w_{ij}^+=\log f_\sss{in}(A_{ij})$ and  $w_{ij}^-=\log f_\sss{out}(A_{ij})$, so that, by Lemma~\ref{lem:correlation-clustering}, it is a projection method with query vector
\begin{equation}\label{eq:query-ppm}
\vec{q}^\text{(PPM)}_{ij}=\log\frac{f_\sss{in}(A_{ij})}{f_\sss{out}(A_{ij})}.
\end{equation}
The standard binary PPM is recovered for $f_\sss{in}(a)=p_\sss{in}^a(1-p_\sss{in})^{1-a}$ and $f_\sss{out}(a)=p_\sss{out}^a(1-p_\sss{out})^{1-a}$.

\subsubsection*{The granularity bias of likelihood maximization.}
It has been observed that likelihood maximization methods for community detection have a bias towards communities of sizes close to $\log n$~\citep{zhang2020statistical,peixoto2021descriptive,gosgens2023correcting}. This can be understood by linking likelihood maximization to Bayesian inference~\citep{peixoto2021descriptive}: Bayesian community detection methods~\citep{peixoto2019bayesian} assume a prior distribution over the set of clusterings and then find the clustering with the highest posterior probability.
Bayes' rule reads
\[
\text{Posterior}(C|A)=\frac{\text{Likelihood}(A|C)\cdot\text{Prior}(C)}{\sum_{C'}\text{Likelihood}(A|C')\cdot\text{Prior}(C')}.
\]
Note that the denominator is constant w.r.t. $C$. If we were to assume a uniform prior, i.e. $\text{Prior}(C)\propto 1$, then we get $\text{Posterior}(C|A)\equiv\text{Likelihood}(A|C)$. This tells us that likelihood maximization is equivalent to Bayesian inference under the assumption of a uniform prior. The uniform distribution over clusterings has been studied in the field of combinatorics for decades~\citep{harper1967stirling,sachkov1997probabilistic}. For example, it is known that, asymptotically as $n\rightarrow\infty$, almost all clusters will have sizes close to $\log n$. This explains why likelihood maximization methods have a bias towards clusterings of this granularity.

\subsection{Which methods are \emph{not} projection methods?}
Not every community detection method fits our hyperspherical framework of community detection. A community detection method is \emph{not} a projection method if the quality function cannot be monotonously transformed to a sum over intra-cluster pairs. Also, in projection methods, the contribution of a node-pair $ij$ to the sum may depend on the input data (e.g., the graph), but \emph{not} on the clustering $C$. In this section, we give a few examples of methods that do not fit this framework.

\subsubsection*{Other inferential methods}
In Section~\ref{sec:likelihood} we saw that some likelihood methods fit our hyperspherical framework. However, not \emph{all} inferential methods are projection methods. For example, suppose we take a PPM where the intra-community density is such that each node has $\lambda$ intra-community neighbors in expectation. Then, if $i$ and $j$ are in a community of size $s$, they should be connected with probability $\lambda/(s-1)$. Hence, the likelihood function $f_\sss{in}$ in~\eqref{eq:query-ppm} would depend on $s$, which requires the query vector $\vec{q}^{\text{(PPM)}}$ to depend on $C$. Our hyperspherical framework does not allow for this.
Similarly, the Bayesian stochastic blockmodeling inference from \cite{peixoto2019bayesian} does not fit our hyperspherical framework because the contribution of each node pair $ij$ depends on the size (and label) of the communities of $i$ and $j$ in an intricate way that cannot be captured by a query vector.

\subsubsection*{$k$-means clustering}
The \emph{$k$-means algorithm} is arguably the oldest and most well-studied clustering method~\citep{jain2010data}. The aim of $k$-means is to divide given vectors $\vec{x}_1,\dots,\vec{x}_n\in\mathbb{R}^{d}$ into $k$ clusters. For each cluster, we compute the \emph{center} as the arithmetic mean of the vectors inside this cluster. The $k$-means algorithm iteratively computes the centers and re-assigns each vector to its nearest center until convergence. 
In~\cite{dhillon2004kernel} it is shown that $k$-means is equivalent to maximizing
\begin{equation}\label{eq:k-means}
    \text{Trace}\left(\hat{H}(C)^\top X \hat{H}(C)\right),
\end{equation}
where $X\in\mathbb{R}^{n\times n}$ is defined by $X_{ij}=\langle\vec{x}_i,\vec{x}_j\rangle$, and $\hat{H}(C)=H(C) (H(C)^\top H(C))^{-1/2}\in\mathbb{R}^{n\times k}$. That is, $\hat{H}(C)_{ia}=s_a^{-1/2}$ if node $i$ is part of the community with label $a$ and size $s_a$. Note that this form resembles the trace-maximization of Markov stability. A straightforward computation shows that  
\[
\left(\hat{H}(C)^\top X \hat{H}(C)\right)_{aa}=\tfrac{1}{s_a}\left(H(C)^\top X H(C)\right)_{aa}.
\]
Therefore, the contribution of each community is again normalized by its size, which is not allowed in our hyperspherical framework.

Another clustering method closely related to $k$-means is \emph{spectral clustering}~\citep{von2007tutorial}. In spectral clustering, we are given an affinity matrix $X\in\mathbb{R}^{n\times n}$ and consider its leading eigenvectors. These leading eigenvectors define coordinates for the $n$ objects, which are then clustered using spatial clustering methods like $k$-means. Spectral clustering differs from other clustering and community detection methods in the sense that it does not explicitly optimize a quality function. However, when using $k$-means for the final clustering step, one could consider it to be optimizing something of the form of~\eqref{eq:k-means}, with $X$ replaced by its low-rank approximation. Therefore, spectral clustering also does not fit our framework of projection methods.

\section{Projection algorithms}\label{sec:algorithms}
The previous section shows that many community detection methods fit our definition of projection methods. A consequence of this is that the same optimization algorithms can be used for each of them. However, it is known that this optimization is NP-hard, and in some forms even APX-hard.

\subsection{Exact optimization}
The general problem of correlation clustering~\citep{bansal2004correlation} and the subproblem of maximizing modularity~\citep{brandes2007modularity,meeks2020parameterised} are known to be NP-complete. However, modularity maximization is known to be Fixed-Parameter Tractable (FPT) when parametrized by the size of the minimum vertex cover of the graph~\citep{meeks2020parameterised}. Nevertheless, it has been shown that modularity maximization~\citep{dinh2015network} and the maximization variant of correlation clustering~\citep{charikar2005clustering} are APX-hard, meaning that approximating it to any constant factor is NP-hard.
This tells us that for large graphs, it may be prohibitively expensive to compute the clustering that minimizes $d_a(\vec{q},\vec{b}(C))$ over the set of clusterings $C$, for a general query vector $\vec{q}$.

Nevertheless, there are some approaches that are able to optimize some of the objectives from Section~\ref{sec:equivalences} with surprising efficiency. In particular, the \emph{Bayan algorithm}~\citep{aref2022bayan} for modularity maximization is able to find the exact modularity maximum in graphs of up to several thousands nodes within hours. This approach relies on an Integer Linear Programming (ILP) formulation. Similar ILP formulations exist for the general problem of correlation clustering~\citep{bansal2004correlation}. These can be converted to the following general ILP formulation of the projection method: in the projection step, we maximize
\[
\sum_{i<j}\vec{q}_{ij}\vec{b}(C)_{ij},
\]
subject to $\vec{b}(C)_{ij}\in\{-1,1\}$ for all $i<j$ and the constraints
\[
    b_{ij}+b_{ik}-b_{jk}\leq1,
\]
for all $i,j,k\in[n]$.

\subsection{Approximate optimization}
There are many approximate maximization algorithms that are able to quickly find clusterings with high modularity. The Louvain~\citep{blondel2008fast} and Leiden~\citep{traag2019louvain} algorithms are perhaps the most well-known heuristics for modularity maximization. \deleted{These algorithms are able to find a local maximum of modularity with running time that (empirically) scales linearly with the number of edges.} 
\added{These algorithms iterate over the nodes and make use of the network sparsity to find the greedy relabeling of a node. For a node $i$, finding this greedy relabeling has complexity $\mathcal{O}(d_i)$.
The Louvain algorithm terminates when it achieves a local maximum. Since it is non-trivial to bound the number of iterations needed to reach a local maximum, there are no theoretical guarantees for the complexity. However, the running time empirically scales linearly with the number of edges.}
While these algorithms often attain values close to the global optimum, they rarely\deleted{ ever} find the exact global optimum~\citep{aref2023heuristic}.

The algorithms proposed in the field of correlation clustering come with theoretical approximation guarantees.
Due to the equivalence established in Lemma~\ref{lem:correlation-clustering}, these algorithms can be applied to modularity maximization. Conversely, modularity maximization algorithms like the Louvain algorithm can be applied to the correlation clustering quality function, allowing for comparisons between these algorithms. Interestingly, while there exist no optimization guarantees for the Louvain algorithm, it does seem to outperform correlation clustering algorithms in such comparisons~\citep{veldt2018correlation}.
The Louvain algorithm can be modified to minimize $d_a(\vec{q},\vec{b}(C))$ with similar performance\replaced{. However,}{, though} the computational complexity may depend on the particular query vector $\vec{q}$~\citep{gosgens2023hyperspherical}.
\added{More precisely, our modification of Louvain assumes a query vector of the form $\vec{q}=\vec{v}(S+L)$, where $S\in\mathbb{R}^{n\times n}$ is a sparse matrix and $L\in\mathbb{R}^{n\times n}$ is a low-rank matrix. 
This way, finding the greedy relabeling for a node $i$ has linear complexity in terms of the number of non-zero entries of $S$ adjacent to $i$. We observe that the running time of this re-implementation of Louvain is proportional to the number of elements in non-zero elements in $S$~\citep{gosgens2023hyperspherical}.}

We denote by $\louv(\vec{q})$ the clustering vector that results from minimizing $d_a(\vec{q},\vec{b}(C))$ over the set of clusterings $C$ by the Louvain algorithm.

\subsection{Do we \emph{need} the global optimum?}
The modularity landscape is known to be \emph{glassy}~\citep{good2010performance}, which means that there are many local maxima with values close to the global maximum. It is likely that for a general query vector $\vec{q}$, the landscape of $d_a(\vec{q},\vec{b}(C))$ suffers from a similar glassiness, which explains why its exact minimization is computationally expensive, while its approximate minimization is computationally cheap.

However, the ultimate goal of community detection is not to minimize the distance to some query vector, but to obtain a meaningful clustering of the network nodes. In settings where we have a generative model, like the PPM, the popular LFR benchmark~\citep{lancichinetti2008benchmark}, or the more recent ABCD benchmark~\citep{kaminski2021artificial}, a meaningful clustering is a clustering that is similar to the planted clustering. However, there is no guarantee that the planted clustering corresponds to the global (or even a local) optimum. 
Most prominently, in sparse network models, it is highly unlikely that  a locally optimal clustering corresponds to the planted clustering. A simple argument for this is that a sparse network model contains isolated nodes with high probability, and these nodes will not be assigned to their true community in any locally optimal clustering. 

Moreover, when applying the Louvain algorithm to graphs from generators, it has been observed that the obtained modularity  often exceeds the modularity of the planted clustering. This tells us that simple greedy optimization algorithms like Louvain already result in a clustering vector $\louv(\vec{q})$ that is nearer to $\vec{q}$ than the planted clustering vector $\vec{b}(T)$, i.e. $d_a(\vec{q},\louv(\vec{q}))\leq d_a(\vec{q},\vec{b}(T))$.

In the experiments for this paper, we have applied the Louvain algorithm 4150 times for different combinations of query mappings and networks. In most of these cases, we chose a query vector using a heuristic, which will be explained in Section~\ref{sec:granularity}, that is designed to ensure $d_a(\vec{q},\louv(\vec{q}))\approx d_a(\vec{q},\vec{b}(T))$. However, in all these 4150 applications of the Louvain algorithm, we have only observed 253 cases where $d_a(\vec{q},\louv(\vec{q}))> d_a(\vec{q},\vec{b}(T))$. Furthermore, most of these might be attributed to numerical errors because there are only 10 instances where $d_a(\vec{q},\vec{b}(C))$ is more than 1\% larger than $d_a(\vec{q},\vec{b}(T))$ and no instances where it is more than 2\% larger.
This tells us that approximate optimization algorithms like Louvain easily find clusterings with higher quality (i.e., lower $d_a(\vec{q},\vec{b}(C))$) than the planted clustering.

The important conclusion from these observations is that better optimization algorithms do not necessarily result in more meaningful clusterings.
Instead, it seems more important to choose the query vector so minimizers of $d_a(\vec{q},\vec{b}(C))$ are close to the ground truth clustering vector $\vec{b}(T)$.

\section{Choosing the query vector: heuristics and open problems}
\label{sec:experiments}
In the previous section, we have seen that the choice of the quality function (or equivalently, the query vector) may have a bigger impact on the performance of a community detection method than the choice of the optimization algorithm.

Choosing a quality function is difficult because it is hard to compare two quality functions in a meaningful way.
However, when restricting to the set of projection methods, the hyperspheric geometry provides us with additional tools to compare the query vectors: for example, we can compute the latitude of the query vector and distances between query vectors. This provides us with some information about the relative position of the query vectors.
In addition, these query vectors define a vector space, so that linear combinations of query vectors also correspond to community detection methods.
In this section, we discuss several ways to choose a query mapping, which maps graphs to a query vectors $\vec{q}$.

\subsection{Graph generators}\label{sec:generators}
We assume that we are given some generator, which produces a tuple $(A,T)$ of an adjacency matrix and a planted clustering (the ground truth). This generator defines a joint distribution on $(A,T)$.
In our experiments, we make use of several different generators:

\noindent\textbf{The Planted Partition Model (PPM).} 
The standard (not generalized) PPM from Section~\ref{sec:likelihood} is the simplest random graph model with community structure.
In this model, there is a planted clustering (partition) of the nodes, and nodes of the same community are more likely to connect to each other than nodes of different communities. We discuss three different variants of the PPM. The first one is a random graph model with homogeneity in both the degree and the community size distribution. We consider $k$ equally sized communities of size $n/k$ (assuming $k$ divides $n$), and assume that each node has (in expectation) the same number of neighbors inside and outside its community, given by $\lambda_\sss{in}$ and $\lambda_\sss{out}$, respectively. We then set the connection probabilities as
\[
p_\sss{in}=\frac{\lambda_\sss{in}}{s-1},\quad\text{and}\quad
p_\sss{out}=\frac{\lambda_\sss{out}}{n-s}.
\]
This way, each node's degree follows the same distribution, which is a sum of two binomially distributed random variables, which can be approximated by a Poisson distribution with \replaced{mean}{rate} $\lambda_\sss{in}+\lambda_\sss{out}$. 

\noindent\textbf{The Heterogeneously-sized PPM (HPPM).}
The second variant of the PPM has homogeneous degrees (again, approximately Poisson distributed), but has heterogeneity in the community\added{-}size distribution. We draw $k$ community sizes from a power-law distribution with some power-law exponent $\delta$, meaning that the probability of obtaining a size $s$ decays as $s^{-\delta}$.
We make sure that each node has on average $\lambda_\sss{in}$ intra-community neighbors and $\lambda_\sss{out}$ neighbors outside of its community, by setting 
\[
p_\sss{in}(s)=\frac{\lambda_\sss{in}}{s-1}
\]
for nodes in communities of size $s$, and 
\[
p_\sss{out}=\frac{n\lambda_\sss{out}}{2\cdot(N-m_T)},
\]
where $m_T$ is the number of intra-community pairs in the planted clustering $T$. 

\noindent\textbf{The Degree-Corrected PPM (DCPPM).} To obtain a graph generator with degree heterogeneity and homogeneous community sizes, we assign \replaced{a weight $\theta_i>0$ to each node}{to each node a weight $\theta_i>0$} and use the PPM parametrization that \added{was} proposed in~\cite{prokhorenkova2019community}. We consider $k$ equally-sized communities of size $n/k$ (again, assuming $k$ divides $n$). We denote the sum of weights inside the $a$-th community by $\Theta_a$ and denote the total weight by $\Theta=\sum_{a=1}^k\Theta_a$. Nodes $i$ and $j$ that are both in the $a$-th community are connected with probability
\[
p_\sss{in}(\theta_i,\theta_j,\Theta_a)=\frac{\lambda_\sss{in}}{\lambda_\sss{in}+\lambda_\sss{out}}\frac{\theta_i\theta_j}{\Theta_a}+\frac{\lambda_\sss{out}}{\lambda_\sss{in}+\lambda_\sss{out}}\frac{\theta_i\theta_j}{\Theta},
\]
and nodes from different communities are connected with probability
\[
p_\sss{out}(\theta_i,\theta_j)=\frac{\lambda_\sss{out}}{\lambda_\sss{in}+\lambda_\sss{out}}\frac{\theta_i\theta_j}{\Theta}.
\]
With these parameters, a node has on average approximately $\lambda_\sss{in}$ neighbors inside its community and $\lambda_\sss{out}$ neighbors outside its community. In addition, the expected degree of a node $i$ is approximately equal to its weight $\theta_i$. To obtain a degree distribution with power-law exponent $\tau$, we draw the weights from a distribution with this same power-law exponent.

\noindent\textbf{The Artificial Benchmark for Community Detection (ABCD).} The Artificial Benchmark for Community Detection (ABCD) is a graph generator that incorporates heterogeneity in both the degree and community\added{-}size distribution in order to generate graphs that resemble real-world networks~\citep{kaminski2021artificial}. This is done by generating a sequence of community sizes and degrees with power-law exponents $\delta$ and $\tau$. Then, it performs a matching process to assign degrees to nodes inside communities. The generator has a parameter $\xi$ that controls the fraction of edges that are inter-community edges.

\noindent\textbf{\added{Parameter choices in our graph generators.}} We set the parameters of the graph generators as follows: we consider graphs with $n=1000$ nodes and mean degree $\lambda_\sss{in}+\lambda_\sss{out}=8$. We choose the parameters of these generators so that each node has (in expectation) $\lambda_\sss{out}=2$ neighbors outside its community. For DCPPM and ABCD, we set the power-law exponent of the degree distribution to $\tau=2.5$.
We generate the planted partitions as follows:
For PPM and DCPPM, we consider $k=50$ communities of size $s=20$ each. For ABCD, we set $\xi=\tfrac{1}{4}$.

\subsection{A heuristic for controlling the granularity of the detected clustering}\label{sec:granularity}
Modularity and Markov stability both have a parameter that controls the granularity of the detected clusterings. Modularity comes with a resolution parameter $\gamma$, and increasing $\gamma$ typically results in detecting communities of smaller sizes. However, it is unclear how this resolution parameter should be chosen in order to detect clusterings of the desired granularity. With `desired', we mean that the granularity of the detected clustering is similar to the granularity of the planted clustering in cases where the graph is drawn from a graph generator. For ER-modularity, there is a particular value of $\gamma(p_\sss{in},p_\sss{out})$ for which maximizing ER-modularity is equivalent to maximizing the likelihood of a PPM with parameters $p_\sss{in},p_\sss{out}$. However, as mentioned in Section~\ref{sec:likelihood}, it is known that maximizing this likelihood is biased towards communities of logarithmic size.
Markov stability comes with a time parameter $t$ which controls the granularity of the detected clustering. Increasing $t$ results in detecting larger communities. Again, it is unclear how this time should be chosen in order to detect communities of the desired granularity.

Within the framework of projection methods, a natural measure of the granularity of a clustering $C$ is the latitude $\ell(\vec{b}(C))$ of the corresponding clustering vector. Hence, in cases where the graph is drawn from a generator with a planted clustering $T$, a clustering with `desired' granularity is a clustering $C$ with $\ell(\vec{b}(C))\approx\ell(\vec{b}(T))$.  In turn, the desired $\ell(\vec{b}(C))$ can be obtained by choosing the right latitude of a query vector. How to make this choice\deleted{,} is the topic of the remainder of this Section~\ref{sec:granularity}. 

The simplest way to change a query vector in order to detect clusterings of coarser granularity, is to add a multiple of $\vec{1}$. That is, a new query vector $\vec{q}'=\vec{q}+c\cdot\vec{1}$ for some $c>0$. The vectors $\vec{q}'$ and $\vec{q}$ lie on the same meridian, i.e., $\dCC(\vec{q},\vec{q}')=0$, while $\vec{q}'$ is further away from $-\vec{1}$, so that $\ell(\vec{q}')>\ell(\vec{q})$. Hence, adding $c\cdot\vec{1}$ is equivalent to projecting the vector $\vec{q}$ to a different latitude, i.e., $\vec{q}'=\parp_\lambda(\vec{q})$ for some $\lambda\in(\ell(\vec{q}),\pi]$. Similarly, the simplest way to change a query vector in order to detect clusterings of finer granularity, is to subtract a multiple \added{of} $\vec{1}$, i.e., $\vec{q}'=\vec{q}-c\cdot \vec{1}$, which is equivalent to $\vec{q}'=\parp_\lambda(\vec{q})$ for $\lambda\in[0,\ell(\vec{q}))$. Hence, the question becomes: given a query vector $\vec{q}$, how should $\lambda$ be chosen such that $\vec{q}'=\parp_\lambda(\vec{q})$ results in a clustering with similar granularity as $\vec{b}(T)$, i.e., such that $\ell\left(\louv(\vec{q}')\right)\approx \ell(\vec{b}(T))$?

In~\cite{gosgens2023correcting}, we proposed a general heuristic that prescribes this latitude as a function of $\lambda_T=\ell(\vec{b}(T))$ and $\theta=\dCC(\vec{q},\vec{b}(T))$. This \emph{granularity heuristic} prescribes
\begin{equation}\label{eq:heuristic}
    \lambda^*(\lambda_T,\theta)=\arccos\left(\frac{\cos\lambda_T\cos\theta}{1+\sin\lambda_T\sin\theta}\right)
\end{equation}
We denote the query vector obtained by applying the granularity  heuristic, by $\vec{q}^*(\vec{q})=\parp_{\lambda^*(\lambda_T,\theta)}(\vec{q})$, where $\vec{q}$ can be any query vector. The latitude in~\eqref{eq:heuristic} is chosen so that $d_a(\vec{q}^*,\vec{b}(T))=\theta$. We have observed that this leads to $d_a(\vec{q}^*,\louv(\vec{q}^*))\approx \theta$ and $\dCC(\vec{q}^*,\louv(\vec{q}^*))\approx\theta$. Whenever these approximations are valid, it can be shown that $\ell(\louv(\vec{q}^*))\approx\ell(\vec{b}(T))$.

\added{We briefly illustrate how solving $d_a(\vec{q}^*,\vec{b}(T))=\theta$ leads to~\eqref{eq:heuristic}: we use~\eqref{eq:correlation-dist} to express $\cos\theta$ in terms of $\lambda_T$, $\lambda^*$ and $d_a(\vec{q}^*,\vec{b}(T))$ like
\[
\cos\theta=\frac{\cos d_a(\vec{q}^*,\vec{b}(T))-\cos\lambda_T\cdot\cos\lambda^*}{\sin\lambda_T\cdot\sin\lambda^*}.
\]
Squaring both sides and making the substitutions $\cos d_a(\vec{q}^*,\vec{b}(T))=\cos\theta$ and $\sin^2\lambda^*=1-\cos^2\lambda^*$ yields
\[
\cos^2\theta=\frac{\cos^2 \theta+\cos^2\lambda_T\cdot\cos^2\lambda^*-2\cos\theta\cdot \cos\lambda_T\cdot\cos\lambda^*}{\sin^2\lambda_T\cdot(1-\cos^2\lambda^*)}.
\]
This can be rewritten to the following quadratic equation in $\cos\lambda^*$:
\[
(1-\sin^2\theta\sin^2\lambda_T)\cos^2\lambda^*-2\cos\theta\cdot \cos\lambda_T\cdot \cos\lambda^*+\cos^2\theta\cdot\cos^2\lambda_T=0,
\]
which has solutions
\begin{align*}
\cos\lambda^*&=\frac{2\cos\theta\cdot \cos\lambda_T\pm \sqrt{4\cos^2\theta\cdot \cos^2\lambda_T-4\cdot (1-\sin^2\theta\sin^2\lambda_T)\cos^2\theta\cdot\cos^2\lambda_T}}{2\cdot (1-\sin^2\theta\sin^2\lambda_T)}\\
&=\cos\theta\cdot \cos\lambda_T\cdot\frac{1\pm \sin\theta\sin\lambda_T}{1-\sin^2\theta\sin^2\lambda_T}\\
&=\frac{\cos\theta\cdot \cos\lambda_T}{1\mp\sin\theta\sin\lambda_T}.
\end{align*}
This gives two possible solutions for $\lambda^*$, one of which corresponds to\eqref{eq:heuristic}. We refer to~\cite{gosgens2023correcting} for the remaining details of the derivation and the experimental validation of this heuristic.
}

Note that $\cos\theta$ is the Pearson correlation coefficient between $\vec{q}$ and $\vec{b}(T)$, and can be considered a measure of how much information $\vec{q}$ carries of the clustering $T$. As a special case, note that $\cos\theta=1$ implies that $\vec{q}$ and $\vec{b}(T)$ lie on the same meridian, and we can see that $\lambda^*(\lambda_T,0)=\lambda_T$, so that $\vec{q}^*=\vec{b}(T)$. In the other extreme, where $\vec{q}$ is not correlated with $\vec{b}(T)$ (i.e., $\theta=\pi/2$), we have $\lambda^*(\lambda_T,\pi/2)=\pi/2$, so that the resulting query vector lies on the equator (just like the modularity vector for $\gamma=1$). For $\theta\in(0,\pi/2)$, the heuristic latitude $\lambda^*(\lambda_T,\theta)$ is between $\lambda_T$ and $\pi/2$.

To compute the heuristic latitude choice in~\eqref{eq:heuristic}, we need estimates of $\lambda_T$ and $\theta$, which requires some knowledge of the planted clustering $T$. It might seem that requiring this knowledge of the planted clustering defeats the purpose of community detection. However, requiring partial knowledge of the planted clustering is not uncommon in other community detection methods, such as likelihood-based methods~\citep{newman2016equivalence,prokhorenkova2019community}. Moreover, when we have access to the graph generator, we can use this to estimate the means of $\lambda_T$ and $\theta$, and use these estimates in~\eqref{eq:heuristic}.

In~\cite{gosgens2023correcting}, we have shown that this granularity heuristic works well for several query vectors $\vec{q}$, including modularity vectors. In this section, we demonstrate that this heuristic also works well for Markov stability vectors.
For $t\in\{1,\dots,5\}$, we consider the projection method with query mapping $\vec{q}^\sss{(MS)}_{t}$ from~\eqref{eq:markov-vector}, which is equivalent to maximizing the Markov stability for a discrete-time random walk with time $t$. We compare this method to the projection method with query mapping $\vec{q}^*\left(\vec{q}^\sss{(MS)}_{t}\right)$, which corresponds to applying our granularity heuristic to the Markov stability vector. 

To quantify the quality of the approximation $\ell(\vec{b}(C))\approx \ell(\vec{b}(T))$, we define the \emph{relative granularity error} as $\ell(\vec{b}(C))/\ell(\vec{b}(T))-1$, which we want to be close to $0$. Positive values indicate that the detected clustering is more coarse-grained than the planted clustering, while negative values indicate that the detected clustering is too fine-grained. 
We measure the similarity between the detected and planted communities by the correlation coefficient $\rho(C,T)=\cos\dCC(\vec{b}(C),\vec{b}(T))$. Values close to $1$ indicate that the clusterings are highly similar, while values close to zero indicate that $C$ is not more similar to $T$ than a random relabeling of $C$.

For the PPM, HPPM, DCPPM and ABCD graph generators, the results are shown in Figures~\ref{fig:Markov-Stability-PPM}, \ref{fig:Markov-Stability-HPPM}, \ref{fig:Markov-Stability-DCPPM} and~\ref{fig:Markov-Stability-ABCD}. For each generator, we generate 50 graphs and show boxplots of the outcomes for each of the query mappings.
\begin{figure}
    \centering
    \begin{minipage}[b]{0.8\textwidth}
        
        \subcaptionbox{Granularity errors of Markov stability on PPM networks.\label{fig:Markov-Stability-PPM-latitude}}{\includegraphics[width=\linewidth]{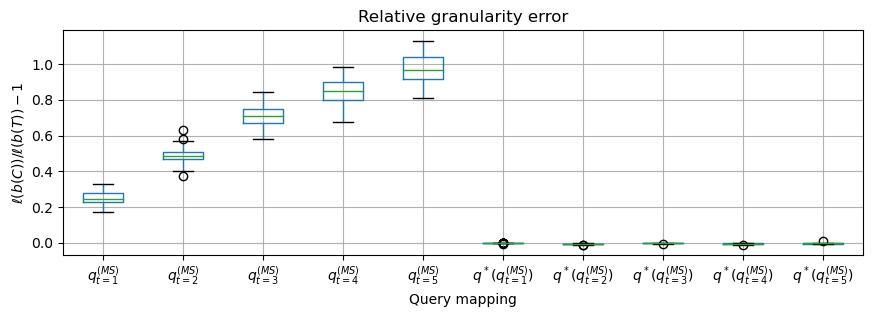}}
        
    \end{minipage}  
   
    \begin{minipage}[b]{0.8\textwidth}
        \subcaptionbox{Performance of Markov stability on PPM networks.\label{fig:Markov-Stability-PPM-performance}}{\includegraphics[width=\linewidth]{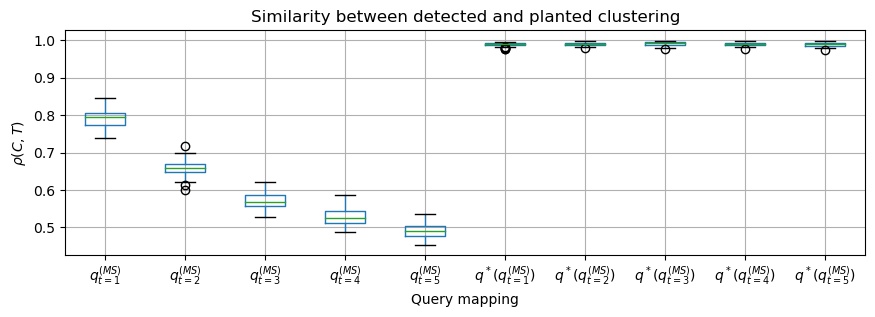}}
    \end{minipage}

    \caption{For 50 graphs drawn from the Planted Partition Model (PPM), we evaluate Markov stability with time $t\in [5]$, and compare the clusterings that are obtained with and without applying the granularity heuristic. $\vec{q}^*(\cdot)$ denotes the heuristic. A positive granularity error indicates that the detected clustering is coarser than the planted clustering. $\rho(C,T)=\cos\dCC(\vec{b}(C),\vec{b}(T))$ is the Pearson correlation between the clustering vectors\added{, which we wish to make close to $1$. We see that $\vec{q}^*(\vec{q}^\sss{(MS)}_t)$} strongly outperforms $\vec{q}^\sss{(MS)}_t$.}
    \label{fig:Markov-Stability-PPM}
\end{figure}

\begin{figure}
    \centering
    \begin{minipage}[b]{0.8\textwidth}
        
        \subcaptionbox{Granularity errors of Markov stability on HPPM networks.\label{fig:Markov-Stability-HPPM-latitude}}{\includegraphics[width=\linewidth]{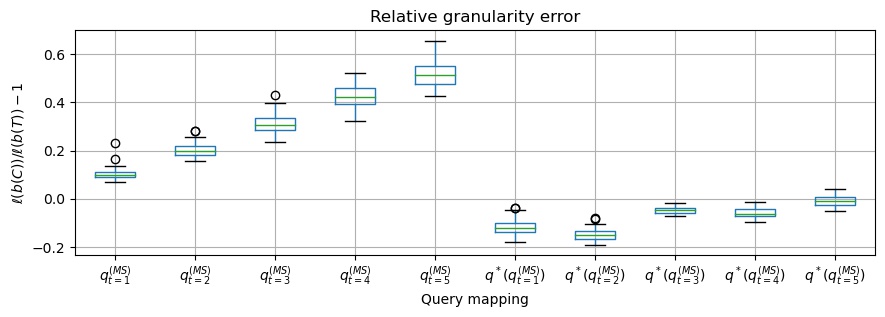}}
        
    \end{minipage}  
   
    \begin{minipage}[b]{0.8\textwidth}
        \subcaptionbox{Performance of Markov stability on HPPM networks.\label{fig:Markov-Stability-HPPM-performance}}{\includegraphics[width=\linewidth]{markov_stability_benchmark_n1000_repeats50_CC_boxplots_PPM.jpg}}
        
    \end{minipage}

    \caption{For 50 graphs drawn from a Heterogeneously-sized Planted Partition Model (HPPM), we evaluate Markov stability with time $t\in [5]$, and compare the clusterings that are obtained with and without applying the granularity heuristic. $\vec{q}^*(\cdot)$ denotes the heuristic. A positive granularity error indicates that the detected clustering is coarser than the planted clustering. $\rho(C,T)=\cos\dCC(\vec{b}(C),\vec{b}(T))$ is the Pearson correlation between the clustering vectors\added{, which we wish to make close to $1$. We see that $\vec{q}^*(\vec{q}^\sss{(MS)}_t)$ strongly outperforms $\vec{q}^\sss{(MS)}_t$}.}
    \label{fig:Markov-Stability-HPPM}
\end{figure}

\begin{figure}
    \centering
    \begin{minipage}[b]{0.8\textwidth}
        \subcaptionbox{Granularity errors of Markov stability on DCPPM networks.\label{fig:Markov-Stability-DCPPM-latitude}}{\includegraphics[width=\linewidth]{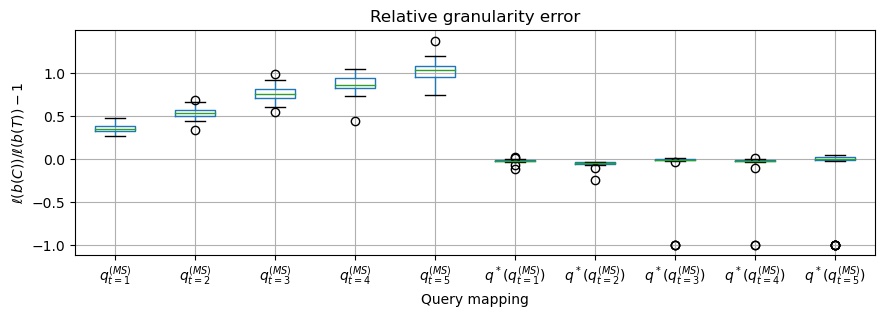}}
    \end{minipage}  
   
    \begin{minipage}[b]{0.8\textwidth}
        \subcaptionbox{Performance of Markov stability on DCPPM networks.\label{fig:Markov-Stability-DCPPM-performance}}{\includegraphics[width=\linewidth]{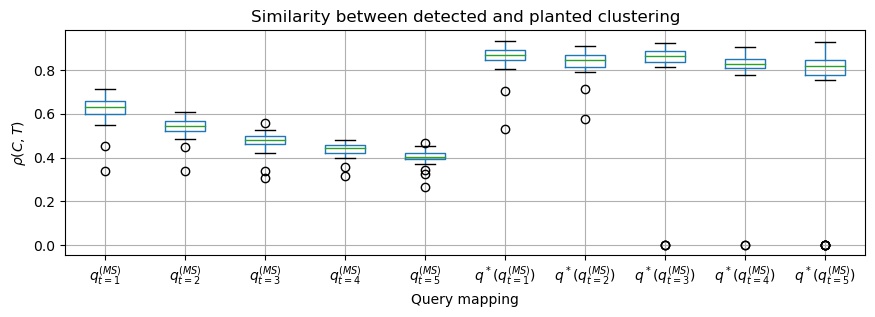}}
        
    \end{minipage}

    \caption{For 50 graphs drawn from the Degree-Corrected Planted Partition Model (DCPPM), we evaluate Markov stability with time $t\in [5]$, and compare the clusterings that are obtained with and without applying the granularity heuristic. $\vec{q}^*(\cdot)$ denotes the heuristic. A positive granularity error indicates that the detected clustering is coarser than the planted clustering. $\rho(C,T)=\cos\dCC(\vec{b}(C),\vec{b}(T))$ is the Pearson correlation between the clustering vectors\added{, which we wish to make close to $1$. We see that $\vec{q}^*(\vec{q}^\sss{(MS)}_t)$ strongly outperforms $\vec{q}^\sss{(MS)}_t$}.}
    \label{fig:Markov-Stability-DCPPM}
\end{figure}

\begin{figure}
    \centering
    \begin{minipage}[b]{0.8\textwidth}
        
        \subcaptionbox{Granularity errors of Markov stability on ABCD networks.\label{fig:Markov-Stability-ABCD-latitude}}{\includegraphics[width=\linewidth]{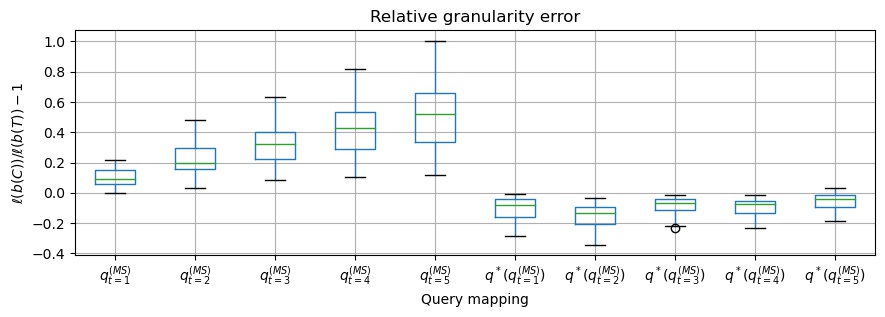}}
        
    \end{minipage}  
   
    \begin{minipage}[b]{0.8\textwidth}
        \subcaptionbox{Performance of Markov stability on ABCD networks.\label{fig:Markov-Stability-ABCD-performance}}{\includegraphics[width=\linewidth]{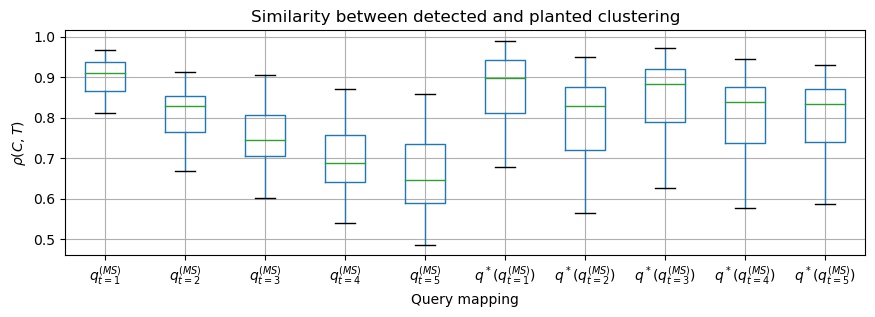}}
        
    \end{minipage}

    \caption{For 50 graphs drawn from the Artificial Benchmark for Community Detection (ABCD), we evaluate Markov stability with time $t\in [5]$, and compare the clusterings that are obtained with and without applying the granularity heuristic. $\vec{q}^*(\cdot)$ denotes the heuristic. A positive granularity error indicates that the detected clustering is coarser than the planted clustering. $\rho(C,T)=\cos\dCC(\vec{b}(C),\vec{b}(T))$ is the Pearson correlation between the clustering vectors\added{, which we wish to make close to $1$. We see that $\vec{q}^*(\vec{q}^\sss{(MS)}_t)$ outperforms $\vec{q}^\sss{(MS)}_t$ for $t\neq 1$}.}
    \label{fig:Markov-Stability-ABCD}
\end{figure}

\noindent\textbf{Effect of the heuristic.}
In Figures~\ref{fig:Markov-Stability-PPM-latitude}, \ref{fig:Markov-Stability-HPPM-latitude}, \ref{fig:Markov-Stability-DCPPM-latitude} and \ref{fig:Markov-Stability-ABCD-latitude}, we see that the granularity heuristic indeed leads to detecting clusterings with granularity closer to the granularity of the planted clustering. For almost all cases, we see that the median relative granularity error after applying the heuristic is closer to zero than before applying the heuristic. The only exception is $t=1$ for the HPPM generator. We see that overall, the granularity heuristic results in clusterings that are slightly more fine-grained than the planted clustering.

In addition, \deleted{the} Figures~\ref{fig:Markov-Stability-PPM-performance}, \ref{fig:Markov-Stability-HPPM-performance}, \ref{fig:Markov-Stability-DCPPM-performance} and \ref{fig:Markov-Stability-ABCD-performance} show that the granularity heuristic typically leads to an increased similarity to the planted clustering. The only two exceptions are HPPM and ABCD for $t=1$, where the granularity heuristic results in slightly lower performance. For the PPM generator, Figure~\ref{fig:Markov-Stability-PPM-performance} shows that for each of the values of $t$, the detection is near-perfect (all similarities are higher than $\rho=0.97$). For the HPPM and DCPPM generators, Figures~\ref{fig:Markov-Stability-HPPM-performance} and~\ref{fig:Markov-Stability-DCPPM-performance} show that the heterogeneity in the community sizes and degrees result in slightly lower performance on these generators.

\noindent\textbf{Markov stability time and granularity.}
Figures~\ref{fig:Markov-Stability-PPM-latitude}, \ref{fig:Markov-Stability-HPPM-latitude}, \ref{fig:Markov-Stability-DCPPM-latitude} and \ref{fig:Markov-Stability-ABCD-latitude} show that for the query vector $\vec{q}^\sss{(MS)}_{t}$ (without applying the heuristic), larger values of $t$ indeed lead to detecting more coarse-grained clusterings. However, $t=1$ already results in clusterings that are more coarse-grained than the planted clustering. We see similar outcomes in Figures~\ref{fig:Markov-Stability-HPPM-latitude}, \ref{fig:Markov-Stability-DCPPM-latitude} and (to a lesser extent) Figure~\ref{fig:Markov-Stability-ABCD-latitude}. Since we are using a discrete-time Markov chain, we cannot consider times $t\in(0,1)$. For a continuous\added{-}time Markov chain, this is possible.

\noindent\textbf{Sparsity and computation time.}
Note that the running time of the Louvain algorithm for Markov Stability vectors depends on the sparsity of the transition matrix $P(t)$. To illustrate: for $t=1$, the transition matrix has the same number of positive entries as the adjacency matrix, and (our implementation of) the Louvain algorithm runs in around $5$ seconds. In contrast, for $t=5$, we have $P(5)_{ij}>0$ for each pair of nodes that is connected by a path of length 5, which leads to a much denser matrix, and the Louvain algorithm takes roughly $200$ seconds per instance. While the continuous-time variant of Markov stability may be able to detect clusterings of finer granularity, the corresponding transition matrix satisfies $P(t)_{ij}>0$ whenever $i$ and $j$ are connected by \emph{any} path. This leads to a much denser matrix, so that the Louvain algorithm will take even more time.

\subsection{Tuning a projection method to a graph generator}\label{sec:training}
In Section~\ref{sec:equivalences}, we have seen that the class of projection methods contains many interesting community detection methods. A further advantage of projection methods is that we can combine different community detection methods by taking linear combinations of their query mappings. On the one hand, this yields infinitely many community detection methods, and gives a lot of flexibility. On the other hand, this begs the question how to choose a suitable query mapping for a task at hand. For instance, are there preferable choices for a particular graph generator? In this section, we will partially address this question by showing how one can \emph{tune} a projection method to a graph generator, in order to maximize the performance on graphs sampled from this generator.

We assume that we are given a generator that produces graphs with their community assignments. We consider linear combinations of a few query mappings and perform a \emph{grid search} to find the coefficients that yield the best query vector. A grid search is a standard approach for hyperparameter tuning, where we discretize each of the parameter intervals and evaluate all possible combinations. To avoid overfitting, one usually generates two different sets of graphs: a \emph{training} set and a \emph{validation} set. The training set is used to find the best hyperparameter combination, while the validation set is used to get an unbiased estimate of the performance of the obtained hyperparameters values.

To demonstrate \added{this method}, we show how we can optimize a projection method for the ABCD graph generator from the previous section. To allow for comparison with this previous section, we consider the same 50 ABCD graphs as the validation set. For the training set, we generate 15 new ABCD graphs from this same generator.

We consider query vectors that are linear combinations of four vectors: the constant vector $\vec{1}$ (to control granularity), the adjacency vector $\vec{v}(A),$ the degree-product vector $\vec{d}(A)$ and the \emph{Jaccard} vector $\vec{j}(A)$. \replaced{The latter is defined}{We define the Jaccard vector} as follows: let $N(i)$ denote the neighborhood of $i$. We follow the convention that $i\in N(i)$ for all $i\in [n]$. Then $\vec{j}(A)_{ij}$ is the Jaccard similarity between the neighborhoods of $i$ and $j$:
\[
\vec{j}(A)_{ij}=\frac{|N(i)\cap N(j)|}{|N(i)\cup N(j)|}.
\]

We consider query vectors that are linear combinations of these four vectors, $\vec{q}=c_1\cdot \vec{1}+c_A\cdot\vec{v}(A)+c_d\cdot\vec{d}(A)+c_j\cdot\vec{j}(A)$. Since in the hyperspherical geometry, the length of the query vector does not affect the detected clustering, we can reduce the number of hyperparameters by one. Assuming that the best combination has $c_A>0$, we set $c_A=1$, thereby making the grid search more efficient. Moreover, for most values of the coefficient $c_1$, the detection method will result in clusterings of a wrong granularity. Thus, instead of fitting $c_1$, we use the granularity heuristic from Section~\ref{sec:granularity}. Then we are left with the following parametrization of query vectors:
\[
\vec{q}(A;c_j,c_d)=\vec{q}^*\left(\vec{v}(A)+c_j\cdot\vec{j}(A)+c_d\cdot\vec{d}(A)\right).
\]
This leaves two hyperparameters to be tuned by the grid search. Note that the vector $\vec{d}(A)$ is a correction term for the degrees. Because of that, the best performance is likely to be found for $c_d\leq0$. We choose the interval $c_d\in[-6,0]$, which we discretize in steps of $\tfrac{1}{2}$.
For the parameter $c_j$, we discretize the interval $[0,1]$ into steps of size $\tfrac{1}{10}$.

The results are shown in Figure~\ref{fig:training}. We see that there is a large region where the method performs well. In particular, the best-performing coefficients are $c_d=-\tfrac{5}{2}$ and $c_j=\tfrac{1}{2}$ with a median performance of $\rho=0.993$ on the training set\footnote{There are four combinations of coefficients that achieve this exact same median performance. We choose this one as it has the highest mean performance.}. We apply this query mapping to the validation set and a median performance of $\rho=0.973$, which is slightly lower than the performance on the training set, as expected due to selection bias. Let us compare this to the performance of the Markov Stability query mappings from Section~\ref{sec:granularity} on these graphs. In Figure~\ref{fig:Markov-Stability-ABCD-performance}, we see that Markov stability with time $t=1$ without applying the heuristic achieves the best median performance on this set of networks. Note that this is equivalent to CL-modularity maximization with $\gamma=1$. This query mapping achieved a median performance of $\rho=0.912$, which is good, but significantly lower than the performance of our optimized query mapping.

\begin{figure}[h!]
    \centering
    \includegraphics[width=0.7\textwidth]{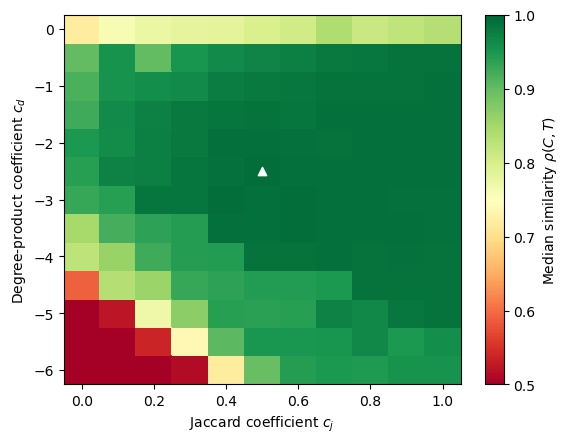}
    \caption{A heatmap of the median performance (similarity between the detected and planted clusterings, as measured by $\rho(C,T)$) for different linear combinations of query mappings. The medians are computed over 10 samples of ABCD graphs, with the same parameters as in the experiments of Section~\ref{sec:granularity}. The best performance is marked with a white triangle.}
    \label{fig:training}
\end{figure}

In this demonstration, we have kept the setup relatively simple by taking combinations of only 4 vectors, reducing this to two coefficients. We did this for simplicity and so that we can visualize the performance on the training set by a two-dimensional heatmap. This already led to strong performance. It is likely that we can improve performance  even further by including a larger number of query vectors, and using a grid search instead of the granularity heuristics to determine the coefficient $c_1$. The obvious downside is that optimization for a larger number of parameters becomes computationally more demanding, and one must increase the size of the training set to avoid overfitting. 

To conclude, we have shown that the class of projection methods unifies many popular community detection methods and is expressive enough to fit realistic benchmark generators like ABCD. This work paves the way to many follow-up research. On the one hand, there are many algorithmic questions: what projection algorithms work well for what query vectors? Are there sets of query vectors for which the minimization of $d_a(\vec{q},\vec{b}(C))$ is \emph{not} NP-hard?
On the other hand, there are also several methodological questions: how do the best-performing coefficients of the linear combination depend on the network properties? For example, how does the \replaced{best-performing coefficient}{coefficient best-performing} $c_d$ depend on the mean and variance of the degree distribution?

\section*{Acknowledgments}
This work was supported by the Netherlands Organisation for Scientific Research (NWO) through the Gravitation {\sc NETWORKS} grant no.\ 024.002.003.

\bibliographystyle{unsrtnat}
\bibliography{main}  






\end{document}

%% file: tikz_preamble.tex
\usepackage{tikz}
\usepackage{etoolbox}
\usetikzlibrary{calc,fadings,decorations.pathreplacing}

\newcommand\pgfmathsinandcos[3]{%
  \pgfmathsetmacro#1{sin(#3)}%
  \pgfmathsetmacro#2{cos(#3)}%
} 
\newcommand\LongitudePlane[3][current plane]{%
  \pgfmathsinandcos\sinEl\cosEl{#2} 
  \pgfmathsinandcos\sint\cost{#3} 
  \tikzset{#1/.estyle={cm={\cost,\sint*\sinEl,0,\cosEl,(0,0)}}}
}
\newcommand\LatitudePlane[3][current plane]{%
  \pgfmathsinandcos\sinEl\cosEl{#2} 
  \pgfmathsinandcos\sint\cost{#3} 
  \pgfmathsetmacro\yshift{\cosEl*\sint}
  \tikzset{#1/.estyle={cm={\cost,0,0,\cost*\sinEl,(0,\yshift)}}} %
}
\newcommand\DrawLongitudeCircle[2][1]{
  \LongitudePlane{\angEl}{#2}
  \tikzset{current plane/.prefix style={scale=#1}}
  \pgfmathsetmacro\angVis{atan(sin(#2)*cos(\angEl)/sin(\angEl))} %
  \draw[current plane] (\angVis:1) arc (\angVis:\angVis+180:1); 
  \draw[current plane,dashed] (\angVis-180:1) arc (\angVis-180:\angVis:1); 
}
\newcommand\DrawLatitudeCircle[2][1]{
  \LatitudePlane{\angEl}{#2}
  \tikzset{current plane/.prefix style={scale=#1}}
  \pgfmathsetmacro\sinVis{sin(#2)/cos(#2)*sin(\angEl)/cos(\angEl)} 
  \pgfmathsetmacro\angVis{asin(min(1,max(\sinVis,-1)))}  
  \draw[current plane] (\angVis:1) arc (\angVis:-\angVis-180:1); 
  \draw[current plane,dashed] (180-\angVis:1) arc (180-\angVis:\angVis:1);
}

\newcommand\projection[5]{ 
    \pgfmathsinandcos\sinEl\cosEl{#1}
    \pgfmathsinandcos\sinlon\coslon{#2}
    \pgfmathsinandcos\sinlat\coslat{#3}
    \pgfmathsetmacro#4{\coslon*\coslat}
    \pgfmathsetmacro#5{\sinlon*\sinEl*\coslat+\sinlat*\cosEl}
}

\newcommand\DrawArcNoMark[9][\R]{
    \pgfmathsinandcos\sinEl\cosEl{#2}
    \pgfmathsinandcos\sinxlon\cosxlon{#3}
    \pgfmathsinandcos\sinxlat\cosxlat{#4}
    \pgfmathsinandcos\sinylon\cosylon{#5}
    \pgfmathsinandcos\sinylat\cosylat{#6}
    \pgfmathsinandcos\sinxylat\cosxylat{#4+#6}
    
    \pgfmathsetmacro\arclen{acos(\sinxlat*\sinxylat+\cosxlat*\cosxylat*\cosylon)}
    \pgfmathsinandcos\sinarclen\cosarclen{\arclen}
    \pgfmathsetmacro\zlon{\ifnumcomp{0}{=}{#4}{#3+90}{atan(-\sinylon/(\sinxlat*(\cosxlat*\sinxylat/\cosxylat-\sinxlat*\cosylon)))}}
    \ifnumcomp{0}{>}{\zlon}{\pgfmathsetmacro\zlon{\zlon+180}}{}
    
    \pgfmathsinandcos\sinzlon\coszlon{\zlon}
    
    \pgfmathsetmacro\zlat{atan(\coszlon*\sinxlat/\cosxlat+\sinzlon*(\sinxylat/\cosxylat-\cosylon*\sinxlat/\cosxlat)/\sinylon)}
    \pgfmathsinandcos\sinzlat\coszlat{\zlat}
    \projection\angEl{\zlon+#3}\zlat\zx\zy
    \projection\angEl{#3}{#4}\startx\starty

    \tikzset{newplane/.estyle={cm={
        \startx,
        \starty,
        \zx,
        \zy,
        (0,0)
    },scale=#1}}
    \draw[newplane,->,line width=0.35mm] (1,0) arc (0:\arclen/2:1) node[below right] {\small #7} arc (\arclen/2:\arclen:1) node[#9 right] {\small #8};
}

\tikzset{%
  >=latex, 
  inner sep=0pt,%
  outer sep=2pt,%
  mark coordinate/.style={inner sep=0pt,outer sep=0pt,minimum size=3pt,fill=black,circle}%
}

%% file: modularity_vectors_drawing.tex
\begin{tikzpicture} 


\def\R{3} 
\def\angEl{5} 
\def\angAz{-101} 
\def\xlong{48}
\def\xlat{25}
\def\bTlat{40}


\pgfmathsetmacro\H{\R*cos(\angEl)} 
\tikzset{xyplane/.style={cm={cos(\angAz),sin(\angAz)*sin(\angEl),-sin(\angAz),
                              cos(\angAz)*sin(\angEl),(0,-\H)}}}
\LongitudePlane[xzplane]{\angEl}{\angAz}
\LongitudePlane[pzplane]{\angEl}{\angAz+\xlong}
\LatitudePlane[equator]{\angEl}{0}


\fill[ball color=white] (0,0) circle (\R); 
\draw (0,0) circle (\R);


\coordinate[coordinate] (O) at (0,0);
\coordinate[mark coordinate] (N) at (0,\H);
\coordinate[mark coordinate] (S) at (0,-\H);
\path[xzplane] (\bTlat:\R) coordinate[mark coordinate] (bT);
\path[xzplane] (0:\R) coordinate[mark coordinate] (q_ER);
\pgfmathsetmacro\cosx{cos(\xlat)};
\pgfmathsetmacro\cosbT{cos(\bTlat)};
\path[xyplane] (\xlong:\R*\cosx) coordinate (zhat);
\path[xyplane] (0:\R*\cosbT) coordinate (xhat);
\path[equator] (\angAz+\xlong/2:\R) coordinate[mark coordinate] (q_CM);

\pgfmathsinandcos\sinzlat\coszlat{15}
\pgfmathsinandcos\sinxlon\cosxlon{\angAz}

\pgfmathsinandcos\sinEl\cosEl{\angEl} 
 \pgfmathsinandcos\sinAz\cosAz{\angAz} 
 \pgfmathsinandcos\sinxlat\cosxlat{\xlat}
 \pgfmathsinandcos\sinxlon\cosxlon{\xlong+\angAz}

\pgfmathsetmacro\perplength{\R/7};


\DrawLatitudeCircle[\R]{0} 
\DrawLongitudeCircle[\R]{\angAz} 




\path (S) +(0.4ex,-0.4ex) node[below] {\small $-\vec{1}$};
\path (bT) node[above left] {\small $\vec{v}(A)$};
\path (q_ER) node[below left] {\small $\vec{q}^\sss{(Mod)}_\sss{ER}$};
\path (q_CM) node[above right] {\small $\vec{q}^\sss{(Mod)}_\sss{CL}$};

\projection\angEl{\angAz-16}{\xlong-90}\labelx\labely;

\pgfmathsetmacro\projectionlat{atan(cos(\xlong)*tan(\bTlat+90))-90-\bTlat};
\pgfmathsetmacro\heuristiclat{acos(cos(\xlong)*cos(\bTlat+90)/(1+sin(\xlong)*sin(\bTlat+90))};
\DrawArcNoMark{\angEl}{\angAz}{\bTlat}{\xlong}{-\bTlat-\bTlat}{}{}{}{};
\draw[xzplane,->,line width=0.35mm] (bT) arc (\bTlat:-\bTlat:\R);
\end{tikzpicture}
    